\documentclass[11pt]{amsart} 

\usepackage[english]{babel}
\usepackage[utf8]{inputenc}
\usepackage[T1]{fontenc}
\usepackage{comment}
\usepackage{amsmath}
\usepackage{amssymb}
\usepackage{amsfonts}
\usepackage{amssymb}
\usepackage{amsthm}
\usepackage{amscd}
\usepackage{xcolor}
\usepackage{slashed}
\usepackage[hidelinks]{hyperref}
\usepackage{geometry}
\usepackage{mathtools}
\usepackage{csquotes}
\usepackage{pgfplots}
\usepackage{float}
\usepackage{cleveref}
\numberwithin{equation}{section}
\usepackage{tikz}
\pgfplotsset{compat=1.16}
\newcommand{\setsymbol}[1]{\ensuremath{\mathbb{#1}}}%
\newcommand{\R}{\setsymbol{R}}%

\newcommand{\A}{\textsc{a}}
\newcommand{\B}{\textsc{b}}
\newcommand{\C}{\textsc{c}}
\newcommand{\D}{\textsc{d}}
\newtheorem{definition}{Definition}[section]
\newtheorem{lemma}[definition]{Lemma}

\newtheorem{theorem}[definition]{Theorem}
\newtheorem{proposition}[definition]{Proposition}
\newtheorem{corollary}[definition]{Corollary}
\newtheorem{example}[definition]{Example}
\newtheorem{remark}[definition]{Remark}

\DeclareMathOperator{\Ric}{Ric}

\title{Initial data sets with vanishing mass are contained in pp-wave spacetimes}

\author{Sven Hirsch}
\address{Columbia University, 2990 Broadway, New York NY 10027, USA}
\email{sven.hirsch@columbia.edu }
\author{Yiyue Zhang}
\address{Beijing Institute of Mathematical Sciences and Applications, Beijing, 101408, China}
\email{zhangyiyue@bimsa.cn}

\begin{document}

\begin{abstract}
In 1981, Schoen-Yau and Witten showed that in General Relativity both the total energy $E$ and the total mass $m$ of an initial data set modeling an isolated gravitational system are non-negative. 
Moreover, if $E=0$, the initial data set must be contained in Minkowski space.
In this paper, we show that if $m=0$, i.e. if $E$ equals the total momentum $|P|$, the initial data set must be contained in a pp-wave spacetime.
Our proof combines spinorial methods with spacetime harmonic functions and works in all dimensions.
Additionally, we find the decay rate threshold where the embedding has to be within Minkowski space and construct non-vacuum initial data sets with $m=0$ in the borderline case.
    As a consequence, this completely settles the rigidity of the spacetime positive mass theorem for spin manifolds. 
\end{abstract}

\maketitle

\section{Introduction}

In Special Relativity, it is well-known that a particle satisfies $E v= P$ where $E$ is the energy, $ v$ the velocity vector, and $ P$ the momentum vector of the particle.
Since nothing moves faster than the speed of light, that is $| v|\le1$, we immediately obtain $E\ge| P|$.
Moreover, we have $E=| P|\ne0$ if and only if $| v|=1$, i.e. when the particle moves at the speed of light.
In this case the particle is a photon (or possibly a gluon or graviton) and corresponds to radiation.
We show that an analogous result holds in General Relativity:

\begin{theorem}\label{T:main}
Let $(M^n,g,k)$, $n\ge3$, be a $C^{2,\alpha}$-asymptotically flat initial data set with decay rate $q\in (\tfrac {n-2}2,n-2]$ satisfying the dominant energy condition.
Suppose that $M^n$ is spin and that $E=|P|\ne0$.
Then $(M^n,g)$ isometrically embeds into a pp-wave spacetime $(\overline M^{n+1},\overline g)$ with second fundamental form $k$.
\end{theorem}

\begin{definition}\label{def: pp-wave spacetimes}
    We say a Lorentzian manifold $(\overline M^{n+1},\overline g)$ is a pp-wave spacetime, or $pp$-wave for short, if $\overline M^{n+1}=\mathbb R^{n+1}$ and
    \begin{align*}
    \overline g=-2du dt+Fdu^2+g_{\mathbb R^{n-1}}
\end{align*}
where $g_{\mathbb R^{n-1}}$ is the flat metric on $\mathbb R^{n-1}$, and $F$ is a function on $\mathbb R^n=\mathbb R^{n-1}\times {{\mathbb{R}_u}}$, independent of $t$, which is superharmonic on $\R^{n-1}\times\{u\}$ for all $u\in \R$, i.e. $\Delta_{\mathbb R^{n-1}} F\le0$, where we use the convention $\Delta=g^{ij}\nabla_{ij}$.
\end{definition}

Such pp-waves are explicit solutions to the Einstein equations and form the simplest model of a gravitational wave.
We provide a detailed overview in Section \ref{S:pp waves overview}.

Since the function $F$ describing a pp-wave is superharmonic, there is a conflict with the notion of asymptotic flatness which forces $F$ to decay sufficiently fast at infinity.
In our next result we show that if certain criteria are met, $F$ must be a constant function, in which case the corresponding pp-wave spacetime is Minkowski space.

\begin{theorem}\label{Cor:main}
Given $\alpha\in (0,1)$, $n\ge3$, and $s\in\mathbb N$ with $s\ge2$, let $(M^n,g,k)$ be a $C^{s,\alpha}$-asymptotically flat spin initial data set with decay rate $q$ which satisfying the dominant energy condition.
      Assume that either 
      \begin{align*}
   E=|P| \quad\quad \text{and}\quad\quad{q>n-1-s-\alpha},
      \end{align*}
      or 
      \begin{align*}
          E=0 \quad\quad\text{and}\quad\quad q>\tfrac{n-2}2.
      \end{align*}
         Then $(M^n,g,k)$ isometrically embeds into Minkowski spacetime with second fundamental form $k$.
\end{theorem}

This has been previously established under various additional assumptions by P.~F.~Yip \cite{Yip}, R.~Beig and P.~Chrusciel \cite{BeigChrusciel}, P.~Chrusciel and D.~Maerten \cite{ChruscielMaerten}, L.-H.~Huang and D.~Lee \cite{HuangLee, HuangLee2}, D.~Kazaras, M.~Khuri and the first named author \cite{HKK}, as well as the authors \cite{HirschZhang}, cf. Remark \ref{rem:previous results}.
We emphasize that in our result no additional decay or regularity assumptions on $\mu,J,g,k$ are made, and that it holds in all dimensions.

The conditions on the decay rate parameter $q$ in Theorem \ref{Cor:main} are optimal.
The inequality $q>\tfrac{n-2}2$ ensures that $E$ and $P$ are well-defined, and $q>n-1-s-\alpha$ is required to rule out non-trivial pp-wave spacetimes:

\begin{theorem}\label{Thm:Example} Let $n\ge3$, $\alpha\in(0,1)$ and suppose that $q:=n-3-\alpha> \frac{n-2}{2}$.
    Then there exists $C^{2,\alpha}$-asymptotically flat initial data set $(M^n,g,k)$ with decay rate $q$ which satisfies $E=|P|\ne0$ and the dominant energy condition.
Moreover, $(M^n,g,k)$ isometrically embeds into a non-trivial pp-wave spacetime with second fundamental form $k$. 
\end{theorem} 

See Section \ref{S:example} for the explicit construction of such initial data sets.
Intriguingly, the asymptotic flatness condition $q> \frac{n-2}{2}$ can only be fulfilled for $n\ge5$, and there are no such initial data sets for $n=3$ and $n=4$.
We point out that in the pioneering work \cite{HuangLee3} by L.-H.~Huang and D.~Lee similar examples were constructed in dimension $n\ge9$.

The proof of Theorem \ref{T:main} combines spinorial and level-set methods.
First, we construct a spacetime harmonic function $u$ where $E\nabla u$ asymptotes to $-P$.
Our goal is to show that the level-sets $\Sigma$ of $u$ are flat which will become the \enquote{planes} of the pp-wave spacetime.
Next, we solve the spacetime Dirac equation $\slashed D\psi=\tfrac12\operatorname{tr} k e_0\psi$ to find a spinor $\psi=(\psi_1,\psi_2)$ which asymptotes to a constant spinor at infinity.
With the help of $u$ and $\psi$, we define in odd dimensions the vector fields
\begin{align*}
    X=\langle e_ie_0 \psi,\psi\rangle e_i,\quad\quad\text{and}\quad\quad Z=|\nabla u|^{-2}\operatorname{Im}\langle e_i(\nabla u) \psi_1,\psi_1\rangle e_i
\end{align*}
where $Z$ is tangential to the level-sets $\Sigma$, and where $\{e_1,\cdots,e_n\}$ is an orthonormal frame of $M^n$.
Using Witten's mass formula, we obtain $X=\nabla u$ and $\nabla^\Sigma Z=0$ which allows us to demonstrate that $\Sigma$ is flat.
A similar argument also works in even dimensions.
Finally, with a delicate analysis involving $u$ and $\psi$, we are able to harness this information to construct the corresponding pp-wave spacetime which will be carried out in Section \ref{S:verification} and Section \ref{S:Killing}.

In Section \ref{S:PDE} we prove Theorem \ref{Cor:main}, and in Section \ref{S:example} we show Theorem \ref{Thm:Example}.
Both proofs are very technical and depend on the precise asymptotics and the underlying PDE for $F$ in a subtle way.

\textbf{Acknowledgements:} SH was supported by the National Science Foundation under Grant No. DMS-1926686, and by the IAS School of Mathematics. YZ was partially supported by NSFC grant NO.12501070 and the startup fund from BIMSA.  The authors thank Hubert~Bray, Simon~Brendle, Piotr~Chru\'sciel, Greg Galloway, Lan-Hsun~Huang, Hyun Chul Jang, Marcus~Khuri, Dan~Lee, and Rick~Schoen for helpful discussions and their interest in this work.
The authors are also grateful to two anonymous referees whose suggestions led to various improvements.

%%%%%%%%%%%%%%%%%%%%%
%%%%%%%%%%%%%%%%%%%%%%%%
%%%%%%%%%%%%%%%%%%%%%%%%%%
%%%%%%%%%%%%%%%%%%%%%%%%%%
%%%%%%%%%%%%%%%%%%%%%%%%%%%%%
%%%%%%%%%%%%%%%%%%%%%%%%%%

\section{Preliminaries and Definitions}
For the sake of completeness, we recall several definitions regarding asymptotically flat manifolds which will be used in the remainder of the paper.
Moreover, we review Witten's proof of the positive mass theorem.

\subsection{Asymptotically flat manifolds}

{Let $(\overline M^{n+1},\overline g)$ be a Lorentzian manifold, let $(M^n,g,k)$ be a smooth spacelike hypersurface in $\overline M^{n+1}$, and let $(\Sigma^{n-1},g_{\Sigma},h)$ be a hypersurface in $M^n$. Throughout the text, we use the following convention for indices:
\begin{itemize}
    \item Greek indices $\alpha,\beta,\gamma,\dots$ for tangent vectors in $\overline M^{n+1}$.
    \item Latin indices $i,j,k,\dots$ for $M^n$, and we denote the normal of $M^n\subseteq \overline{M}^{n+1}$ with $e_0$.
    \item Small Latin caps $\textsc{a,b,c},\dots$ for $\Sigma^{n-1}$, and we denote the normal of $\Sigma^{n-1}\subseteq M^n$ with ${\hat{\mathbf{n}}}$.
\end{itemize}}

\begin{definition}[Weighted function spaces] \label{Def:weighted}
{Let $B\subseteq \mathbb R^n$ be a ball containing the origin. }
For  $\alpha\in(0,1)$, $s\in \mathbb N$, $p\in[1,\infty)$ and $q\in \R$, the corresponding weighted H\"older norm for functions, tensors and spinors is defined by
    \begin{equation*}
        \|f\|_{C^{s,\alpha}_{-q}(\R^n\setminus B)}:=
    \sum_{|I|\le s}\left||x|^{|I|+q}\nabla_{I} f\right|+\sum_{|I|=s} \sup_{\substack{x,y\in \R^n\setminus B 
    \\ |x-y|\le \frac{|x|}{2}}}|x|^{\alpha+k+q}\frac{|\nabla_I f(x)-\nabla_I f(y)|}{|x-y|^\alpha}.
    \end{equation*}
Similarly, the weighted Sobolev norm is given by
    \begin{equation*}
        \|f\|_{W^{s,p}_{-q}(\R^n\setminus B)}=\left(\int_{\R^n\setminus B } \sum_{|I|\le s}\left||x|^{|I|+q}|\nabla_I f| 
\right|^p|x|^{-n}dx\right)^\frac{1}{p}.
    \end{equation*}
\end{definition}

\begin{definition}[Asymptotically flat initial data sets] \label{Def:AF}
Let $(M^n,g)$ be a connected complete Riemannian manifold without boundary, let $k$ be a symmetric $2$-tensor on $M^n$, and let $s\in\mathbb N$ with $s\ge 2$.
We say $(M^n, g, k)$ is a $C^{s,\alpha}$-asymptotically flat initial data set of decay rate $q\in (\tfrac {n-2}2,n-2]$ if it satisfies the following conditions:  

There is a compact set $\mathcal{C}\subset M$ such that we can write $M\setminus \mathcal{C}=\cup_{\ell=1}^{\ell_0}M_{end}^{\ell}$ where the ends $M_{end}^\ell$ are pairwise disjoint and admit diffeomorphisms $\phi^\ell$ to the complement $\mathbb{R}^n \setminus B$ of a ball.
Moreover, on each end, we have
\begin{equation}\label{asymflat}
(\phi_\ast^\ell g- g_{\mathbb R^n},\phi_\ast^\ell k)\in C^{s,\alpha}_{-q}{(\mathbb R^n\setminus B)}\times C^{s-1,\alpha}_{-q-1}{(\mathbb R^n\setminus B)}, 
\end{equation}
where $g_{\R^n}$ is the Euclidean metric. %, and $C^{2,\alpha}_{-q}$ and $C^{2,\alpha}_{-q}$ are weighted H\"older spaces. 
The energy density $\mu$ and the momentum density $J$ satisfy
\begin{equation*}
    \begin{split}
        \mu:=&\frac{1}{2}\left( R+(\operatorname{tr}_gk)^2-|k|_g^2\right) \in L^1(M^n),
        \\ J:=&\operatorname{div}(k-(\operatorname{tr}_gk)g) \in L^1(M^n).
    \end{split}
\end{equation*}
\end{definition}

We sometimes also say that $(M^n,g,k)$ is $C^{s,\alpha}_{-q}$-asymptotically flat.
The definition of weighted function spaces from Definition \ref{Def:weighted} extends to asymptotically flat initial data sets with the help of the diffeomorphisms $\phi^\ell$.

\begin{remark}
For notational convenience, we do not allow $(M^n,g,k)$ to have an interior boundary $\partial M^n$. However, it is easy to see that all our results also hold true for asymptotically flat manifolds with boundaries as long as the spinorial positive mass theorem is still valid. 
In particular, we can allow interior MOTS and MITS boundaries with $g|_{\partial M^n}\in C^{1,\alpha}$ and $k|_{\partial M^n}\in C^{0,\alpha}$.
We describe the minor adjustments necessary in Section \ref{SS:boundary}.
\end{remark}

\begin{definition}[Dominant energy condition]
  We say $(M,g,k)$ satisfies the \emph{dominant energy condition} if $\mu\ge |J|$.
\end{definition}

\begin{definition}[ADM energy, momentum and mass]
    The ADM energy $E$ and the ADM momentum $P$ are defined as by
    \begin{equation}
    \begin{split}
E=&\frac{1}{2(n-1)\omega_{n-1}}\lim_{r\to \infty}\int_{|x|=r}(g_{ij,j}-g_{ii,j}){\hat{\mathbf{n}}}^j dA,
\\ P_i=& \frac{1}{(n-1)\omega_{n-1}}\lim_{r\to \infty}\int_{|x|=r}(k_{ij}-(\operatorname{tr}_g k)g_{ij}){\hat{\mathbf{n}}}^j dA,
    \end{split}
    \end{equation}
    where $\omega_{n-1}$ is the volume of a unit $n-1$ dimensional sphere, and ${\hat{\mathbf{n}}}$ is the outer unit normal to the sphere $\{|x|=r|\}$.
    Moreover, in case $E\ge|P|$, the ADM mass is defined by $m=\sqrt{E^2-|P|^2}$.
\end{definition}

\subsection{Witten's proof of the positive mass theorem}

Let $(M^n,g,k)$ be an asymptotically flat initial data set which is spin.
We denote with $\mathcal S$ the spinor bundle of $M^n$, with $\nabla$ the induced connection, and with $\slashed D=e_i\nabla_i$ the Dirac operator.
\begin{definition}
   We say that $\overline{\mathcal S}=\mathcal S\oplus \mathcal S$ is the spacetime spinor bundle. Given $\phi=(\phi_1,\phi_2)\in \overline{\mathcal S}$, $\R^{n,1}=\operatorname{span}\{e_0,e_1,\dots,e_n\}$ acts on $\overline{\mathcal S}$ via Clifford multiplication 
   \begin{align}
    e_0(\phi_1,\phi_2)=(\phi_2,\phi_1)\quad\quad\text{and}\quad\quad e_i(\phi_1,\phi_2)=(e_i\phi_1,-e_i\phi_2).
   \end{align}
   The corresponding connection and Dirac operator will still be denoted with $\nabla, \slashed D$.
\end{definition}

\begin{theorem}\label{T:Witten}
Let $M^n$ be spin and let $(M^n,g,k)$ be a $C^{s,\alpha}$-asymptotically flat initial data set of order $q\in(\tfrac{n-2}{2},n-2]$ satisfying the dominant energy condition.
Then for every constant spinor $\psi^\infty$ in the designated end, there exists a spinor field $\psi\in\overline{\mathcal S}(M^n)$ satisfying $\slashed D\psi=\frac12\operatorname{tr}_g(k)e_0\psi$ such that $(\psi-\psi^\infty)\in  C^{s,\alpha}_{-q}$ for any $\alpha\in(0,1)$ in the designated end, and $\psi\in C^{s,\alpha}_{-q}$ in all other ends.
Moreover,
\begin{align*}
    E|\psi^\infty|^2+\langle \psi^\infty,Pe_0\psi^\infty\rangle=\frac{2}{(n-1)\omega_{n-1}}\int_{M^n}\left(\left|\nabla \psi +\frac{1}{2}k_{\cdot j}e_je_0\psi\right|^2+\frac12\mu|\psi|^2+\frac12\langle \psi, J e_0\psi\rangle\right)dV.
\end{align*}
\end{theorem}

\begin{proof}
        The proof of the integral formula and existence of $\psi\in W^{1,2}_{-q}$ is well-established, see for instance \cite[Corollary 8.26]{Lee} and \cite[Proposition 8.21]{Lee}.
    Since the regularity $(\psi-\psi^\infty)\in C^{s,\alpha}_{-q}$ is often not stated precisely in the literature, we give a brief sketch of an argument here.
    The spacetime Dirac operator $\slashed D-\frac12\operatorname{tr}_g(k)e_0$ is a linear elliptic system of first order. 
    Hence, we may write the equation in local coordinates and apply the Calderon-Zygmund estimates for elliptic systems, see for instance \cite[Theorem 6.2.5]{Morrey}.
    This yields $W^{1,p}_{-q}$ regularity for any $p>1$ which can be bootstrapped to $W^{2,p}_{-q}$. 
    Using Sobolev embedding and Schauder estimates for elliptic systems, finishes the proof.
\end{proof}

\section{Gravitational waves}\label{S:pp waves overview}

Gravitational waves are one of the most striking predictions of General Relativity. 
While the wave-like nature of the Einstein equations\footnote{In the sense that their linearization at Minkowski space is the wave equation.} was already discovered in 1916 by A.~Einstein himself \cite{Einstein1, Einstein2}, it took another century till LIGO \cite{LIGO} was able to experimentally observe gravitational waves in 2016.
The simplest explicit example of such a gravitational wave is a pp-wave spacetime $(\overline M^{n+1},\overline g)$, which is short for a plane-fronted wave with parallel rays.
This goes back to H.~Brinkmann's seminal (though purely mathematical) work from 1925 \cite{Brinkmann1925}.
Recall from the introduction that in this case
$\overline M^{n+1}=\mathbb R^{n+1}$ and
    \begin{align*}
    \overline g=-2du dt+Fdu^2+g_{\mathbb R^{n-1}},
\end{align*}
where the \emph{wave profile function} $F:\mathbb R^n=\mathbb R^{n-1}\times {{\mathbb{R}_u}}\to\mathbb R$ is independent of $t$ and superharmonic on $\R^{n-1}\times\{u\}$ for all $u\in \R$.
They play a crucial role in physics, serving as fundamental models for gravitational radiation.
A certain kind of pp-waves, called \emph{plane wave spacetimes}, appears as \enquote{tangent spaces} of null geodesics in any Lorentzian manifolds; see \cite{Blau2, Blau3, Penrose} for a discussion of such \emph{Penrose limits}.

\subsection{Examples of pp-waves}

As mentioned in the introduction $F\equiv 1$ corresponds to Minkowski space. 
In fact, a more general statement holds:

\begin{example}
Let $F=F(u)$ be a function depending only on $u$. Then
\begin{align*}
    \overline g=-2dudt+Fdu^2+g_{\mathbb R^{n-1}}=du(-2dt+Fdu)+\sum_{\A=1}^{n-1}y_\A^2=dudv+\sum_{\A=1}^{n-1}y_\A^2,
\end{align*}
where $v=-2t+\int_0^uF(s)ds$.
Hence, the corresponding pp-wave is Minkowski space.
\end{example}

The following important example is due to L.-H.~Huang and D.~Lee \cite[Chapter 2]{HuangLee3}:

\begin{example}\label{Ex HL}
    Let $F(\mathbf{y},u)=1+\kappa(\mathbf{y})\eta(u)$, where 
    {{$\mathbf{y}=(y_1,\dots, y_{n-1})\in \R^{n-1}$,}} 
    $\kappa:\mathbb R^{n-1}\to\mathbb R$ is a fixed positive superharmonic function such that $\kappa=\mathcal O(|\mathbf{y}|^{-(n-3)})$ near $\infty$, and let $\eta:\mathbb R\to\mathbb R$ be a cutoff function, i.e.,  ${\operatorname{supp}(\eta)\subseteq [-1,1]}$. 
    Then $\Delta_{\mathbb R^{n-1}}F\le0 $ for each $u$, and $F$ gives rise to a pp-wave spacetime.
    Moreover, the $t=0$ slice of this spacetime is asymptotically flat for $n\ge9$.
\end{example}

Here $(\kappa-1)=\mathcal O(|\mathbf{y}|^{-(n-3)})$ denotes that $(\kappa-1)$ decays like $|\mathbf{y}|^{-(n-3)}$ as $|\mathbf{y}|\to\infty$. {Later we will also need some more precise notions to differentiate between different types of decay which will be introduced in Section \ref{S:Killing}.}

To see where the dimensional restriction comes from, note that $(F-1)=\mathcal O(|\mathbf{y}|^{-(n-3)})$.
Taking derivatives in the  $u$-direction, we also obtain $\partial_u\partial_u F=\mathcal O(|\mathbf{y}|^{-(n-3)})$ (note that there is no decay improvement).
However, asymptotic flatness requires $(g-\delta)=\mathcal O(|\mathbf{y}|^{-q})$ and $\partial_u\partial_u g\in \mathcal O(|\mathbf{y}|^{-q-2})$ for $q>\frac{n-2}2$. 
Comparing exponents, we find that $n\ge9$ is needed to ensure $F$ is non-trivial.

When the above initial data sets are asymptotically flat, their energy and momentum are well-defined and we have $E=|P|\ne0$.
Hence, Huang-Lee's spacetime gives a striking counterexample to the previous widely believed statement that an IDS with zero mass (i.e. $E=|P|$) must be contained in Minkowski space. 
The identity $E=|P|$ can be verified using spinorial methods which will be carried out in a more general setting below.

In Figure \ref{f1}, we have visualized the pp-wave from Example \ref{Ex HL}; also see \cite[Figure 1]{HirschJangZhang} by H.J.~Jang and the authors for a similar depiction.

\begin{figure}
\centering

\begin{tikzpicture}
    \begin{axis}[
        axis lines = middle,
        xlabel = {\(x_1, \dots, x_{n-1}\)},
        ylabel = {\(x_n\)},
        zlabel = {\(t\)},
        xlabel style = {anchor=north east, xshift=65pt, yshift=10pt},
        ylabel style = {anchor=north west, xshift=-15pt, yshift=0pt},
        zlabel style = {anchor=south},
        domain = -2:2,
        samples = 20,
        samples y = 20,
        zmin = 0, zmax = 4,
        xmin = -.5, xmax = 2,
        ymin = -.4, ymax = 1.5,
        xtick = \empty,
        ytick = \empty,
          ztick = {0,  2,  4},
        zticklabels = {0, 1, 2},
        enlargelimits = true,
        view = {70}{20} 
    ]
    \addplot3 [
        fill=orange,
        fill opacity=0.2,
        draw=none
    ] coordinates {
        (0, -.5, -.5)
        (0, 1.5, 3.5)
        (0, 1.5, 4.5)
        (0, -.5, 0.5)
    } -- cycle;

    \addplot3 [
        fill=red,
        fill opacity=0.2,
        thick,
        line width=.25mm
    ] coordinates {
        (0, 0, 0)
        (1, 0, 0)
        (1, .5, 0)
        (0, .5, 0)
    } -- cycle;
    \addplot3 [
        thick,
        orange,
        domain=-.5:1.5,
        samples=2
    ] 
    ({0}, {x}, {2*x+0.5 }); 

    \addplot3 [
        thick,
        orange,
        domain=-.5:1.5,
        samples=2
    ] 
    ({0}, {x}, {2*x + 1.5});

        \addplot3 [
        thick,
        red,
        domain=0:5,
        samples=2
    ] 
    ({.5}, {.25}, {x});
    \end{axis}
\end{tikzpicture}
\qquad
\begin{tikzpicture}

    \fill[red!20] (1,4) rectangle (2,5);
    \draw[thick] (1,4) rectangle (2,5);
   \fill[red!20] (1,.5) rectangle (2,1.5);
       \draw[thick] (1,.5) rectangle (2,1.5);
 \coordinate (A) at (-.5,2.25);
    \coordinate (B) at (3.5,2.25);
    \coordinate (E) at (3.5,2.75);
    \coordinate (C) at (3,3.25);
    \coordinate (D) at (0,3.25);
    \coordinate (F) at (-.5,2.75);

    \draw[thick, fill=red!20] (A) -- (B) -- (E) -- (C) -- (D) -- (F) -- cycle;
    \node at (1.5, 1.75) {\textbf{$t = 2$}};
\node at (1.5, 5.25) {\textbf{$t = 0$}};
\node at (1.5, 3.5) {\textbf{$t = 1$}};

\end{tikzpicture}
\caption{
The majority of the spacetime in Example \ref{Ex HL} is vacuum (which corresponds to $\eta=0$) and coincides with Minkowski space, with the exception of the wave itself ($\eta\ne0$), which is highlighted as an orange beam moving at the speed of light. This beam extends in the \(x_1, \dots, x_{n-1}\) directions with an appropriate fall-off towards $\infty$.
To understand a pp-wave's impact, observe its effect on an observer, marked by the red line.  As the pp-wave passes through the observer (around time \(t=1\)), a notable elongation occurs in the \(x_n\) direction.
This stretching effect is strongest near the center of the wave, and diminishes for large \(x_1, \dots, x_{n-1}\). After the wave has passed through the observer, everything returns to its original state. 
}
\label{f1}
\end{figure}

\subsection{Geometric properties of pp-waves and their initial data sets}

The proof of the following result can be found in \cite[Chapter 2]{HuangLee3} by L.-H.~Huang and D.~Lee:

\begin{theorem}\label{3 thm 1}
Let $(\overline M^{n+1},\overline g)$ be a pp-wave given by a function $F$.
    Then $(\overline M^{n+1},\overline g)$ is a null dust spacetime.
    More precisely, the Einstein tensor $\overline G=\overline \Ric-\frac12\overline R\overline g$ satisfies
\begin{align*}
    \overline G=v\otimes v
\end{align*}
where the velocity vector field $v$ is null and given by
\begin{align*}
    v=\sqrt\mu |\nabla u|^{-1}\frac{\partial}{\partial t}.
\end{align*}
    In particular, $(\overline M^{n+1},\overline g)$ satisfies the spacetime dominant energy condition.
    Moreover, for any initial data set $(M^n,g,k)$, we have
    \begin{align*}
        \mu=-\frac12F^{-1}\Delta_{\mathbb R^{n-1}}F.
    \end{align*}
\end{theorem}

This also explains the superharmonicity of $F$ in Definition \ref{def: pp-wave spacetimes}, as it ensures that $\mu\ge0$ which is needed for the dominant energy condition to be satisfied.

\begin{theorem}\label{thm pp waves ids identities}
    Let $(M^n,g,k)$ be the ($t=-\ell$)-graph in a pp-wave spacetime $(\overline M^{n+1},\overline g)$ with wave profile function $F$.
    Then 
  \begin{equation}\label{g IDS S3}
    g=(F+2\ell_u)du^2+2\sum_{\A=1}^{n-1}\nabla^\Sigma_\A \ell dudy_\A+\sum_{\A=1}^{n-1}dy_\A^{2}. 
\end{equation}
where $\{dy_\A\}$ forms an orthonormal basis of $\Sigma=\mathbb R^{n-1}$.
\end{theorem}

\begin{proof}
The identity for the metric simply follows by plugging $t=-\ell$ into the spacetime metric $\overline g=-2dtdu+Fdu^2+g_{\R^{n-1}}$, where  $g_{\R^{n-1}}=\sum_{\A=1}^{n-1}dy_\A^{2}$.
\end{proof}

In Section \ref{S:Killing}, specifically {see \eqref{eq:kAB} and \eqref{eq:barg}}, we will establish further useful identities for such initial data sets, including
\begin{align*}
    F=&|\nabla u|^{-2}+|\nabla^\Sigma\ell|^2-2\ell_u,\\
     k_{\A\B}   =&|\nabla u|\nabla^\Sigma_{\A\B}\ell.
\end{align*}

In order to prove Theorem \ref{T:main}, we will show that the metric $g$ of an arbitrary initial data set with $E=|P|$ has the form given in \eqref{g IDS S3}.
Doing so leads to several technical complications.
In particular, we do not know a priori what the quantities $u,F,\ell,y_\A,\Sigma$ are.

\begin{theorem}\label{thm pp wave parallel spinor}
Every initial data set $(M^n,g,k)$ contained in a pp-wave $(\overline M^{n+1},\overline g)$ admits a spinor $\psi\in \overline{\mathcal S}$ solving $\nabla_i\psi=-\frac12k_{ij}e_je_0\psi$,  where $e_0$ is the normal vector to $M^n$ in $\overline M^{n+1}$.
Moreover, on each IDS there is a spacetime harmonic function $u$ satisfying $\nabla_{ij}u=-k_{ij}|\nabla u|$ with flat level-sets.
\end{theorem}

\begin{proof}
First, we note that $u$ is covariantly constant in $(\overline M^{n+1},\overline g)$, and that $\overline \nabla u$ is a null vector. Hence $\overline \nabla_{ij}u=0$ which implies $\nabla_{ij}u=-k_{ij}e_0(u)=-k_{ij}|\nabla u|$; also see \cite{HKK, BHKKZ}.
Moreover, by the construction of the pp-wave, $u$ has flat level-sets.

It is well-known that each pp-wave $(\overline M^{n+1},\overline g)$ admits a paralle spinor $\psi$, see for instance \cite{Aichelburg, Araneda, Bryant}. Restricting $\psi$ to $(M^n,g,k)$, we obtain $\nabla_i\psi=-\frac12k_{ij}e_je_0\psi$, where $e_0$ denotes the normal vector of $(M^n,g,k)$ within $(\overline M^{n+1},\overline g)$. 
We give another, more explicit proof below.

Let $\{e_\A\}$ be an orthonormal basis on $\mathbb R^{n-1}$ and let ${\hat{\mathbf{n}}}$ be the normal of $\mathbb R^{n-1}$ within the $(t=0)$-slice of $(\overline{M}^{n+1},\overline g)$.
Let $\phi$ be a parallel spinor on $\mathbb R^{n-1}$, which we note is constant with respect to the frame $\{e_\A\}$.
Next, we extend $\phi$ to a spinor on $M^n=\mathbb R^{n-1}\times\mathbb R$ by prescribing $\phi$ to be constant along the frame $\{e_\A,{\hat{\mathbf{n}}}\}$.
This uses that $\mathcal S(\mathbb R^{n-1})\subseteq \mathcal S(\mathbb R^n)$.
Since $\phi$ is constant with respect to the frame $\{e_\A,{\hat{\mathbf{n}}}\}$, we may use \cite[Theorem 4.14]{LawsonMichelson} to compute
\begin{equation*}
\begin{split}
      \nabla_{{\hat{\mathbf{n}}}} \phi=&\frac{1}{2}\sum_{i<j}\langle\nabla_ {{\hat{\mathbf{n}}}} e_i, e_j\rangle e_i e_j\phi
      \\=& \frac{1}{2}\left(\sum_{1\le\A<\B<n}\langle\nabla_ {\hat{\mathbf{n}}} e_\A, e_\B\rangle e_\A e_\B\phi+\sum_{\A}\langle\nabla_ {\hat{\mathbf{n}}} e_\A, {\hat{\mathbf{n}}}\rangle e_\A {\hat{\mathbf{n}}}\phi\right).%     =\frac{1}{2}k_{{\hat{\mathbf{n}}}\A}e_\A {\hat{\mathbf{n}}} \phi.
\end{split}
\end{equation*}
Next, we observe that 
$${\langle \nabla_{\hat{\mathbf{n}}}e_\A, {\hat{\mathbf{n}}}\rangle=\nabla_{\hat{\mathbf{n}}}\langle e_\A,{\hat{\mathbf{n}}}\rangle-\langle e_\A,\nabla_{\hat{\mathbf{n}}} {\hat{\mathbf{n}}}\rangle=-\langle e_\A ,\nabla_{\hat{\mathbf{n}}}(|\nabla u|^{-1}\nabla u)\rangle= k_{{\hat{\mathbf{n}}}\A}}.$$
Moreover, recall 
\begin{align*}
      \overline{g}=-2d(t+\ell) du+g
\end{align*}
and
\begin{equation*}
    g=(F+2l_u)du^2+2\sum_{\A=1}^{n-1}\nabla^\Sigma_\A \ell dudy_\A+\sum_{\A=1}^{n-1}dy_\A^{2},
\end{equation*}
which implies that $e_\A=\partial_{y_\A}$ and ${\hat{\mathbf{n}}}= |\nabla u|(\partial_u-\partial_\A \ell \partial_\A)$.
Therefore, we  have 
\begin{align*}
     [{\hat{\mathbf{n}}},e_\B]=-|\nabla u|^{-1}(e_\B |\nabla u|){\hat{\mathbf{n}}}+|\nabla u|(\nabla^\Sigma_{\A\B}\ell) e_\A
\end{align*}
which yields
\begin{equation*}
    \begin{split}
        \langle\nabla_{\hat{\mathbf{n}}} e_\A, e_\B\rangle=&\frac{1}{2}\left(\langle[e_\B,e_\A],{\hat{\mathbf{n}}}\rangle+\langle[e_\B, {\hat{\mathbf{n}}}],e_\A\rangle-\langle[e_\A,{\hat{\mathbf{n}}}],e_\B\rangle\right)=0.
    \end{split}
\end{equation*}
Consequently, $$ \nabla_{{\hat{\mathbf{n}}}} \phi=\frac{1}{2}k_{{\hat{\mathbf{n}}}\A}e_\A {\hat{\mathbf{n}}} \phi.$$
Therefore, the spinor $\psi_1=|\nabla u|^{\tfrac{1}{2}}\phi$ satisfies
\begin{equation*}
    \nabla_{{\hat{\mathbf{n}}}} \psi_1=\frac{1}{2}k_{{\hat{\mathbf{n}}} i}e_i {\hat{\mathbf{n}}}\psi_1.
\end{equation*}
Moreover, according to \eqref{psi1},  we have $\nabla_\A \psi_1=\frac{1}{2}k_{\alpha i}e_i {\hat{\mathbf{n}}}\psi_1$.
We now define the spinor $\psi=(\psi_1,-{\hat{\mathbf{n}}} \psi_1)$ on the spacetime spinor bundle $\overline{\mathcal S}=\mathcal S\oplus \mathcal S$.
Then we have $\nabla_i\psi=-\frac12k_{ij}e_je_0\psi$ where we recall that $e_0(\phi_1,\phi_2)=(\phi_2,\phi_1)$ for any spinor $\phi=(\phi_1,\phi_2)\in \overline{\mathcal S}$.
\end{proof}

Note that the above argument leads to not just one, but many parallel spinors on pp-wave spacetimes: one for each parallel spinor on $\mathbb R^{n-1}$.

\begin{lemma}\label{l3.6}
Let $(M^n,g,k)$ be an initial data set admitting a spinor satisfying $\nabla_i\psi=-\frac12k_{ij}e_je_0\psi$.
   Then  $\mu |\psi|^2=-\langle Je_0\psi,\psi\rangle$ for the energy and momentum densities $\mu$ and $J$.
\end{lemma}

\begin{proof}
    We have 
    \begin{align*}
     -\nabla^\ast\nabla\psi=&   \nabla_i\nabla_i\psi\\
        =&-\frac12\nabla_i(k_{ij}e_je_0\psi)\\
        =&-\frac12(\nabla_ik_{ij})e_je_0\psi+\frac14k_{ij}k_{il}e_je_0e_le_0\psi\\
        =&        -\frac12(\nabla_ik_{ij})e_je_0\psi+\frac14|k|^2\psi
    \end{align*}
    and 
    \begin{align*}
    \slashed D^2\psi=&    e_i\nabla_i(e_j\nabla_j\psi)\\
        =&       -\frac12 e_i\nabla_i(e_j k_{jk}e_ke_0\psi)\\
       =&\frac12 e_i\nabla_i(\operatorname{tr}_gk e_0\psi)\\
       =&\frac12e_i(\nabla_i \operatorname{tr}_gk) e_0\psi-\frac14e_i(\operatorname{tr}_gk) e_0 k_{ij}e_je_0\psi\\
       =&\frac12e_i(\nabla_i \operatorname{tr}_gk) e_0\psi-\frac14 (\operatorname{tr}_gk)^2\psi.
    \end{align*}
    Combining this with the Schr\"odinger-Lichnerowicz formula $\slashed D^2=\nabla^\ast\nabla +\frac14R$, the result follows.
\end{proof}

\begin{corollary}\label{3 thm 2}
    Let $(M^n,g,k)$ be an asymptotically flat initial data set admitting a spinor solving $\nabla_i\psi=-\frac12k_{ij}e_je_0\psi$.
    Suppose that $\mu\ge|J|$.
    Then $E=|P|$ and $m=\sqrt{E^2-|P|^2}=0$.
    In particular, every asymptotically flat slice within a pp-wave has vanishing mass.
\end{corollary}

\begin{proof}
   Recall that in the proof of Theorem \ref{thm pp wave parallel spinor}, we have established $\psi=(\psi_1,-{\hat{\mathbf{n}}} \psi_1)$.
   Hence $X_i:=\langle e_ie_0\psi,\psi\rangle=|\psi|^2{\hat{\mathbf{n}}}$. 
  Therefore, Witten's integral formula, Theorem \ref{T:Witten}, yields
  \begin{align*}
      E|\psi^\infty|^2-|P||\psi^\infty|^2 = \frac2{(n-1)\omega_{n-1}} \int_M \left (\left |\nabla_i\psi+\frac12k_{ij}e_je_0\psi\right |^2+\frac12\mu|\psi|^2+\frac12\langle Je_0\psi,\psi\rangle\right)dV
  \end{align*}
  where $\psi^\infty$ is the constant spinor at $\infty $ the spinor $\psi$ asymptotes to.
  Using Lemma \ref{l3.6} above and the assumptions on $\nabla \psi$, the result follows.
\end{proof}

For a further discussion of pp-waves and a historical overview, we refer to the survey \cite{AazamiCederbaum} by A.~Azami, C.~Cederbaum and C.~Roche.

\section{Construction of a spacetime harmonic function}\label{S: spacetime harmonic}

In this section, we construct a spacetime harmonic function with flat level-sets.
This is the main geometric ingredient in the proof of Theorem \ref{T:main}.
The level-sets of the spacetime harmonic function correspond to the \enquote{planes} of the pp-wave.
We begin with a quick summary of the proof idea:

If $\psi$ solves $\nabla_i\psi=-\frac12k_{ij}e_je_0\psi$, then the vector field $X_i=\langle e_ie_0\psi,\psi\rangle$ is closed and therefore locally gives rise to a function $u$ with $\nabla u=X$.
It turns out that $u$ is spacetime harmonic and that its level-sets $\Sigma$ are MOTS with the second fundamental form $h_{ij}=-k_{ij}$.
Moreover, since $\nabla u$ is non-vanishing and the level-sets of $u$ are $(n-1)$-dimensional planes near infinity, we can deduce that $M^{n}=\mathbb R^{n-1}\times\mathbb R=\mathbb R^{n}$ topologically.
While $\psi=(\psi_1,\psi_2)$ itself is not parallel on $\Sigma$, we can use it to construct a parallel spinor $\phi:=|\psi|^{-1}\psi_1$ on $\Sigma$.
This allows us to obtain the flatness of the level-sets $\Sigma$.

We summarize the precise results of this section below:

\begin{theorem}\label{T:spacetime harmonic}
Let $s\ge2$, $s\in \mathbb N$, and let $\alpha\in(0,1)$.
    Suppose $(M^n,g,k)$ is a $C^{s,\alpha}$-asymptotically flat spin initial data set with decay rate $q\in(\tfrac{n-2}{2},n-2]$ satisfying the dominant energy condition.
    Moreover, assume that $E=|P|$.
    Then the following holds true:
    \begin{enumerate}
         \item There exists a spacetime harmonic function $u-u_{\infty}\in C^{s+1,\alpha}_{1-q}$ satisfying the Hessian equation $\nabla^2u=-k|\nabla u|$, where $u_{\infty}=-|P|^{-1}P\cdot x$ if $|P|\neq 0$, and $u_{\infty}=x_n$ if $|P|=0$.
        \item Topologically, $M^n=\R^n$.
        \item There exists a constant $c$ such that $ |\nabla u|\ge c$ everywhere.
        \item The level-sets of $u$ are flat, i.e. the Riemann curvature tensor of the induced metric on the level-sets is vanishing.
        \item The second fundamental form $h$ of the level sets $\Sigma $ satisfies $h=-k|_{T\Sigma\otimes T\Sigma}$.
        \item We have $J=-\mu {\hat{\mathbf{n}}}$, where ${\hat{\mathbf{n}}}=|\nabla u|^{-1}\nabla u$ is the unit normal to the level-sets.
    \end{enumerate}
\end{theorem}

\subsection{Existence of a spacetime harmonic function}

In case $P\ne0$, we define $\mathbf p= P|P|^{-1}$, and without loss of generality $\mathbf p{=(0,\dots,0,-1)\in \mathbb{R}^n}$.
In case $P=0$, {{we set $\mathbf p=(0,\dots,0,-1)$ as well}}.

Let $\psi_1^\infty$ be any constant spinor in $\mathcal{S}$ with the norm $\|\psi_1^\infty\|=\frac{\sqrt{2}}{2}$, and let $\psi^\infty=(\psi^\infty_1,\mathbf{p}\psi^\infty_1)$ be a unit constant spinor in $\overline{\mathcal{S}}$.
According to Theorem \ref{T:Witten}, there exists a spinor $\psi\in C^{s,\alpha}(M^n)$ solving $\slashed D\psi=\frac12(\operatorname{tr}_gk)e_0\psi$, which asymptotes to $\psi^\infty$.
Since 
\begin{align*}
    \langle Pe_0\psi^\infty,\psi^\infty\rangle=\langle P\mathbf{p}\psi^\infty_1,\psi^\infty_1\rangle-\langle P\psi^\infty_1,\mathbf{p}\psi^\infty_1\rangle=-2|P||\psi^\infty_1|^2=-|E||\psi^\infty|^2,
\end{align*}
the spinor $\psi$ must also satisfy 
\begin{align*}
    \nabla_i\psi=-\frac12k_{ij}e_je_0\psi,
\end{align*}
according to Witten's mass formula.

\begin{lemma}
Following \cite[Appendix B]{BeigChrusciel}, we define 
\begin{align*}
    X_i=\langle e_ie_0\psi,\psi\rangle,\quad\quad\text{and}\quad\quad
    N=|\psi|^2.
\end{align*}
Then $X$ and $N$ are differentiable, and we have
\begin{align}\label{eq:X N identity}
    \nabla_iX_j=-k_{ij}N,\quad\quad \text{and} \quad\quad
    \nabla_i N=-k_{ij}X_j.
\end{align}
\end{lemma}
\begin{proof}
We compute
\begin{align*}
    \begin{split}
        \nabla_iX_j=&
        \langle e_je_0 \nabla_i\psi,\psi\rangle+\langle e_je_0\psi,\nabla_i\psi\rangle\\
        =& -\frac12\langle e_je_0 k_{il}e_le_0\psi,\psi\rangle-\frac12\langle e_je_0\psi,k_{il}e_le_0\psi\rangle\\
         =& \frac12k_{il}\langle e_j e_l\psi,\psi\rangle-\frac12k_{il}\langle e_j\psi,e_l\psi\rangle\\
         =&-k_{ij}|\psi|^2
    \end{split}
\end{align*}
and
\begin{align*}
    \begin{split}
        \nabla_i N=&\langle \nabla_i\psi,\psi\rangle+\langle \psi,\nabla_i\psi\rangle\\
        =&-\frac12\langle k_{ij}e_je_0\psi,\psi\rangle-\frac12\langle\psi,k_{ij}e_je_0\psi\rangle\\
        =&-k_{ij}\langle e_je_0\psi,\psi\rangle.
    \end{split}
\end{align*}
This establishes identity \eqref{eq:X N identity}.
\end{proof}

\begin{corollary}
    On each simply connected set $\Omega^n\subset M^n$, there exists a function $u$ satisfying $du=X$.
\end{corollary}
\begin{proof}
    Since $\nabla_iX_j=-k_{ij}N=\nabla_jX_i$, this follows from the Poincar\'e lemma. 
\end{proof}
\begin{lemma}\label{L:asymptotics}
    We have 
\begin{align*}
    X+\mathbf{p}\in C^{s,\alpha}_{-q},\quad\quad N-1\in C^{s,\alpha}_{-q}
\end{align*}
in the designated asymptotically flat end, and
\begin{align*}
    X\in C^{s,\alpha}_{-q},\quad\quad N\in C^{s,\alpha}_{-q}
\end{align*}
in all other ends.
\end{lemma}

\begin{proof}
We compute
\begin{align*}
    N-1=&|\psi|^2-|\psi^\infty|^2=\frac12(\langle \psi-\psi^\infty,\psi+\psi^\infty\rangle+\langle \psi+\psi^\infty,\psi-\psi^\infty\rangle).
\end{align*}
Thus, the result follows from Theorem \ref{T:Witten}.
Moreover, for $\psi:=(\psi_1,\psi_2)$ which asymptotes to $(\psi_1^\infty,\mathbf{p}\psi_1^\infty)$ at infinity, we have
\begin{align*}
\begin{split}
    X_i=&\langle e_i(\psi_2,\psi_1),(\psi_1,\psi_2)\rangle\\
    =&-2\operatorname{Re}\langle e_i\psi_1,\psi_2\rangle\\
    =&-2\operatorname{Re}(
\langle e_i\psi^\infty_1,\mathbf{p}\psi^\infty_1\rangle+\langle e_i(\psi_1-\psi^\infty_1),\mathbf{p}\psi^\infty_1\rangle+\langle e_i\psi_1,\psi_2-\mathbf{p}\psi^\infty_1\rangle
    ).
    \end{split}
\end{align*}
Since $2\operatorname{Re}\langle e_i\psi^\infty_1,\mathbf{p}\psi^\infty_1\rangle=-\mathbf{p}_i$ in the designated asymptotically flat end, and $2\operatorname{Re}\langle e_i\psi^\infty_1,\mathbf{p}\psi^\infty_1\rangle=0$ in all other ends, the result follows again from Theorem \ref{T:Witten}. 
\end{proof}

\begin{lemma}
    We have $N=|X|$.
\end{lemma}

\begin{proof}
{{
Using the gradient equations for $X$ and $N$ in \eqref{eq:X N identity}, we have 
\begin{equation*}
    \nabla_i|X|^2=2\langle \nabla_i X_j,X_j\rangle=-2k_{ij}X_j N=\nabla_i N^2.
\end{equation*}
}}
Since both $N$ and $|X|$ asymptote to $1$ at $\infty$, the result follows.
\end{proof}

\begin{proof}[Proof of Theorem \ref{T:spacetime harmonic} (1)]
    The above lemma immediately implies the existence of a spacetime harmonic function $u\in C^{s}$ on each simply connected domain $\Omega^n\subset M^n$, satisfying the PDE $\nabla^2u=-k|\nabla u|$.
    The decay estimates follow analogously to \cite[Section 4]{HKK}, where they have been established for spacetime harmonic functions in dimension 3.
    Finally, the improved regularity, $u\in C^{s+1,\alpha}_{1-q}$, will follow from the fact that $|\nabla u|\ge c$ for some constant $c$ which will be shown in Lemma \ref{L:gradient bound} below.
    Thus, the proof of item (1) will be complete once we have established that $M=\mathbb R^n$.
    This will be done in the following section.
\end{proof}

\subsection{$M$ is topologically trivial}

In dimension $3$, cf. \cite[Proposition 7.2]{HKK}, 
the resolution of the Poincar\'e conjecture and the geometrization conjecture is used to show that $M\cong\R^3$. 
Hence, a new argument is needed for the higher dimensional case.

\begin{lemma}\label{L:gradient bound}
    The vector field $X$ is non-vanishing and there exists a constant $c>0$ depending only on $(M^n,g,k)$ such that $c\le|X|\le c^{-1}$ everywhere on $M^n$.
    In particular, item (3) of Theorem \ref{T:spacetime harmonic} holds in case $X$ is globally integrable.
\end{lemma}

\begin{proof}
    The argument of \cite[Proposition 7.1]{HKK} still goes through even if $X$ is not a priori  globally integrable.
More precisely, for any point $p\in M^n$, take a point $q$ in the designated asymptotically flat end with $|X(q)|\ge\tfrac12$. 
Connecting $p$ with $q$ via a geodesic $\gamma$ and observe that $|\nabla_{\dot\gamma}|X||=|k(X,\dot \gamma)|\le |k||X|$. 
    Integrating this ODE and using the decay assumption $k\in C^{s-1,\alpha}_{-1-q}$, the result follows.
\end{proof}

\begin{lemma}
    The second fundamental form of the level-sets of $u$ satisfies $h=- k|_{T\Sigma\otimes T\Sigma}$, i.e. item (5) of Theorem \ref{T:spacetime harmonic} holds.
\end{lemma}

\begin{proof}
    On the one hand, we have $\nabla_{ij}u=-k_{ij}|\nabla u|$. On the other hand, $\nabla_{\A\B}u=\nabla_{\A\B}^\Sigma u+h_{\A\B}|\nabla u|$ for tangential $e_\A,e_\B$.
    Thus, the result follows.
\end{proof}

\begin{corollary}
    The initial data set $(M^n,g,k)$ has just a single asymptotically flat end.
\end{corollary}

\begin{proof}
    Recall that by construction $|X|$ asymptotes to $1$ in the designated asymptotically flat end, and asymptotes to $0$ in all other ends. 
    However, the latter would contradict the previous lemma.
\end{proof}

\begin{proof}[Proof of Theorem \ref{T:spacetime harmonic} (2)]
This follows directly from Reeb's global stability theorem \cite[Theorem 3.1, p.112]{Godbillon}.
For completeness, we have included an entire proof adjusted to our setting in Appendix \ref{S:topology}.
\end{proof}

\subsection{The flatness of the level-sets}

Next, we proceed with the geometric core of our argument, and show in this section that the Riemann curvature tensor of the level-sets $\Sigma_t=\{u=t\}$ vanishes.
In dimension 3, this follows from combining the spacetime Hessian equation $\nabla^2u=-k|\nabla u|$, Bochner's formula and the identity $\mu|\nabla u|=-\langle J,\nabla u\rangle$, cf. \cite[Corollary 2.2]{HirschZhang}.
However, in higher dimensions this chain of thoughts only yields that the scalar curvature of the level-sets vanishes.
Hence, a better argument is needed.

\begin{proof}[Proof of Theorem \ref{T:spacetime harmonic} (4)]
Recall that for any constant spinor $\psi^\infty_1$, we defined $\psi^\infty=(\psi^\infty_1,\mathbf{p}\psi^\infty_1)$ where $\mathbf{p}=|P|^{-1}P$.
Moreover, there exists a spinor $\psi=(\psi_1,\psi_2)$ which is asymptotic to $\psi^\infty$ and satisfies
\begin{align*}
    \nabla_i\psi=-\frac12k_{ij}e_je_0\psi.
\end{align*}
This implies
\begin{align*}
    \nabla_i\psi_1=-\frac12k_{ij}e_j\psi_2\quad\quad\text{and}\quad\quad \nabla_i\psi_2=\frac12 k_{ij}e_j\psi_1.
\end{align*}
Denote with ${\hat{\mathbf{n}}}=\frac{\nabla u}{|\nabla u|}$ the unit normal to the level-sets $\Sigma$ which is well-defined since $|\nabla u|\ne0$.
We have
\begin{align*} 
    \begin{split}
        |\nabla u|=\langle X,{\hat{\mathbf{n}}}\rangle=\langle{\hat{\mathbf{n}}} e_0\psi,\psi \rangle= \langle {\hat{\mathbf{n}}}\psi_2,\psi_1\rangle-\langle {\hat{\mathbf{n}}} \psi_1,\psi_2\rangle\le 2|\psi_1||\psi_2|\le |\psi|^2=|\nabla u|
    \end{split}
\end{align*}
Therefore, 
\begin{equation} \label{psi12}
\psi_1={\hat{\mathbf{n}}}\psi_2\quad\quad \text{and}\quad\quad \psi_2=-{\hat{\mathbf{n}}}\psi_1,     
\end{equation}
 as well as
\begin{align}\label{psi1 identity}
     \nabla_i \psi_1=\frac{1}{2}k_{ij}e_j{\hat{\mathbf{n}}}\psi_1\quad\quad\text{and}\quad\quad \nabla_i\psi_2=\frac{1}{2}k_{ij}e_j{\hat{\mathbf{n}}}\psi_2.
\end{align}
If the dimension $n$ is odd, then the spinor bundle on $M^n$ can be identified as the spinor bundle on the level sets $\Sigma$.  Recall that $h_{\A\B}=\langle\nabla_\A{\hat{\mathbf{n}}},e_\B \rangle=-k_{\A\B}$, then
\begin{equation}
    \begin{split} \label{psi1}
    \nabla_\A^{\Sigma}\psi_1
   =&\nabla_\A \psi_1+\frac{1}{2}h_{\A \B}e_\B {\hat{\mathbf{n}}} \psi_1
   \\=&\frac12 k_{\A j}e_j{\hat{\mathbf{n}}}\psi_1-\frac12k_{\A\B}e_\B {\hat{\mathbf{n}}}\psi_1
   \\=&-\frac{1}{2}k_{\A{\hat{\mathbf{n}}}}\psi_1
   \\=&\frac{1}{2}(|\nabla u|^{-1}\nabla_\A |\nabla u|)\psi_1
    \end{split}
\end{equation} 
Hence, $\nabla^{\Sigma}(|\nabla u|^{-\frac{1}{2}}\psi_1)=0$, i.e., $|\nabla u|^{-\frac{1}{2}}\psi_1$ is a parallel spinor on $\Sigma$.

If $n$ is even,  then the spinor bundle on the level sets $\Sigma$ can be identified as an eigenspace of the linear transformation $\sigma: \mathcal{S}(M^n)\to \mathcal{S}(M^n)$, where $\sigma=i^\frac{n}{2}e_1\cdots e_n$, also see the discussion in \cite[p.903]{harmonicspinor}. Note that $\sigma^2=1$ and $\sigma$ commutes with the even part of $\operatorname{Cl}(\mathbb{R}^n)$, denote as
$\operatorname{Cl}^0(\mathbb{R}^n)$, which is generated by $e_je_l$.
Therefore, $\mathcal{S}(M^n)=\mathcal{S}^+(M^n)\oplus \mathcal{S}^-(M^n)$, where $\mathcal{S}^+(M^n)$ and $ \mathcal{S}^-(M^n)$ are  $\pm 1$ eigenspaces of $\sigma$. 
 The isomorphism  $\operatorname{Cl}(\R^{n-1})\to \operatorname{Cl}^0(\R^n)$ induced by $e_\A\to e_\A e_n$ gives a $\operatorname{Cl}(\R^{n-1})$ module structure on $\mathcal{S}^+(M^n)$. Thus,  $\mathcal{S}(\Sigma)\cong \mathcal{S}^+(M^n)|_\Sigma\cong \mathcal{S}^-(M^n)|_\Sigma$, and the odd part of  $\operatorname{Cl}(\R^n)$, denote as $\operatorname{Cl}^1(\R^n)$, maps $\mathcal{S}^+(M^n)$ to $\mathcal{S}^-(M^n)$. 
Without loss of generality, we choose the $+1$ eigenspace of $\sigma$.  Therefore, for any $\phi\in \mathcal{S}(M^n)$, $(1+\sigma)\phi\in \mathcal{S}^+(M^n)$ because $\sigma(1+\sigma)\phi=(1+\sigma)\phi$. Then using $\nabla \sigma=\sigma \nabla $, Equation \eqref{psi12} implies 
\begin{equation*}
    \nabla_i(1+\sigma)\psi_1=\frac{1}{2}k_{ij}e_j{\hat{\mathbf{n}}}(1+\sigma)\psi_1.
\end{equation*}
Thus, $|\nabla u|^{-\tfrac{1}{2}}(1+\sigma)\psi_1$ is a parallel spinor on $\Sigma$. 

Since the asymptotic of $\psi_1$ was chosen arbitrary, we have an abundance of parallel spinors.
These spinors give rise to parallel vector fields $Z^\B$:
we define $(Z^\B)_\A=|\nabla u|^{-1}\operatorname{Im}\langle e_\A{\hat{\mathbf{n}}} \psi_1,\psi_1\rangle$ in odd dimensions. Similarly, we define  $(Z^\B)_\A=|\nabla u|^{-1}\operatorname{Im}\langle e_\A{\hat{\mathbf{n}}} (1+\sigma) \psi_1 ,(1+\sigma)\psi_1\rangle$ in even dimensions. In both cases, we can choose $\psi_1$ such that $Z^\B$ are asymptotic to the coordinate vector fields $e_\B$ on $\Sigma$.
At $\infty$, $\{Z^\B\}$ forms an orthonormal basis on $T\Sigma$.
Since 
\begin{align*}
    \nabla_\A^\Sigma\langle Z^\B,Z^\C\rangle=0,
\end{align*}
$\{Z^\beta\}$ are linearly independent everywhere on $\Sigma$.
Therefore, $\Sigma$ is flat.
\end{proof}

\begin{remark}
    Alternatively, flatness also follows directly from the fact that the spinors $|\nabla u|^{-\tfrac{1}{2}}\psi_1$ or $|\nabla u|^{-\tfrac{1}{2}}(1+\sigma)\psi_1$ are parallel for any choice of spinor $\psi_1^\infty$.
\end{remark}

\begin{proof}[Proof of Theorem \ref{T:spacetime harmonic} (6)]
    Witten's formula \ref{T:Witten} implies together with the condition $\mu\ge|J|$ that
    \begin{align*} 
    \mu|\psi|^2+\langle Je_0\psi,\psi\rangle=0.
\end{align*}
Inserting the defintions of $X=\nabla u$, $N=|\nabla u|$, the result follows.
\end{proof}

\subsection{Adjustments for non-empty interior boundary}\label{SS:boundary}

In case we have an interior MOTS or MITS boundary, we solve the PDE $\slashed D\psi=0$, ${\hat{\mathbf{n}}} e_0\psi =-\psi $ on MOTS, ${\hat{\mathbf{n}}} e_0 \psi = \psi$ on MITS, and imposing the same boundary conditions at infinity.
Then Witten's mass formula with boundary, cf. \cite{Hawking}, \cite[Theorem 8.29]{Lee} and \cite[Theorem 11.4]{BartnikChrusciel}, yields again that $\psi$ satisfies $\nabla_i\psi=\frac12k_{ij}e_je_0\psi$.
Now all parts of the proof of Theorem \ref{T:spacetime harmonic} go through verbatim with the exception of item (2).
The boundary condition for $\psi$ implies that $X$ is normal to  the inner boundaries, i.e.,
\begin{align*}
    X_{\hat{\mathbf{n}}}=&\langle {\hat{\mathbf{n}}} e_0\psi,\psi\rangle =\pm |\psi|^2=\pm|X|.
\end{align*}
Hence, the inner boundaries are leaves of the foliation. By Reeb’s global stability theorem (cf.\ref{Reeb thm with boundary}), the leaves are homeomorphic to $\mathbb{R}^{n-1}$, yielding a contradiction.

\section{Verification of the Gauss and Codazzi equations}\label{S:verification}

Recall that small Latin capitals indicate tangential indices (to $\Sigma$), Roman letters denote arbitrary indices on $M^n$, and the normal vector to $\Sigma\subseteq M^n$ is represented by ${\hat{\mathbf{n}}}$.
Following \cite{HirschZhang}, we use the notation 
\begin{equation*}
    \begin{split}
        \overline{R}_{ijkl}=& R_{ijkl}+k_{il}k_{jk}-k_{ik}k_{jl},
        \\ A_{ijk}=&\nabla_ik_{jk}-\nabla_{j}k_{ik}.
    \end{split}
\end{equation*}
If $\overline R$ and $A$ are identically zero, $(M^n,g,k)$ isometrically embeds into Minkowski space with second fundamental form $k$ by the Lorentzian version of the fundamental theorem of hypersurfaces.
In our setting we are able to show that most $\overline R$ and $A$ terms are vanishing.
These identities will become useful in the next section when we construct the Killing development of $(M^n,g,k)$.

\begin{lemma}
    We have
    \begin{align*}
        \overline{R}_{\A\B\C\D}=0.
    \end{align*}
\end{lemma}

\begin{proof}
    This follows immediately from the Gauss equations:
    \[{{R_{\A\B\C\D}=R^\Sigma_{\A\B\C\D}+h_{\A\C}h_{\B\D}-h_{\A\D}h_{\B\C}}},\] 
    combined with the facts that $R^\Sigma_{\A\B\C\D}=0$ and $h_{\A\B}=-k_{\A\B}$ from Theorem \ref{T:spacetime harmonic}.
\end{proof}

\begin{lemma}\label{Lemma:Codazzi}
We have 
\begin{equation}\label{Codazzi}
     \text{(i)}\,\,  \overline R_{ij\A{\hat{\mathbf{n}}}}=0,\quad\quad  \text{(ii)}\,\,   A_{\A\B\C}=0, \quad\quad \text{and} \quad\quad \text{(iii)}\,\, A_{i\B{\hat{\mathbf{n}}}}=0.
\end{equation}
\end{lemma}
\begin{proof}
We first show that
        \begin{align}\label{eq:R=-A}
        \overline R_{\A\B\C{\hat{\mathbf{n}}}}=-A_{\A\B\C}.
    \end{align}
For this purpose, we compute
\begin{align*}
\begin{split}
    A_{\A\B\C}=&\nabla_{\A}k_{\B\C}-\nabla_{\B}k_{\A\C}\\
    =&\nabla^\Sigma_\A k_{\B\C}-h_{\A\B}k_{\C{\hat{\mathbf{n}}}}-h_{\A\C}k_{\B{\hat{\mathbf{n}}}}-\nabla^\Sigma_\B k_{\A \C}+h_{\B\A}k_{\C{\hat{\mathbf{n}}}}+h_{\B\C}k_{\A{\hat{\mathbf{n}}}}
    \end{split}
\end{align*}
Now, using the fact $h=-k|_\Sigma$, and the Codazzi equations $\nabla^\Sigma_\A h_{\B\C}-\nabla^\Sigma_\B h_{\A\C}=R_{\A\B\C{\hat{\mathbf{n}}}}$, claim \eqref{eq:R=-A} follows.
Hence, it suffices to show item (i) and (iii) of \eqref{Codazzi}. 

Next, we use the identities $\nabla_i\psi_1=\frac12k_{ij}e_j{\hat{\mathbf{n}}} \psi_1$, cf. Equation \eqref{psi1 identity}. and $\nabla_\A\nabla_\B \psi-\nabla_\B\nabla_\A\psi=\frac14 R_{\A\B ij}e_ie_j\psi$, to obtain
    \begin{align*}
        \begin{split}
            0=&\nabla_\A(2\nabla_\B\psi_1-k_{\B i}e_i{\hat{\mathbf{n}}}\psi_1)-\nabla_\B(2\nabla_\A\psi_1-k_{\A i}e_i{\hat{\mathbf{n}}}\psi_1)\\
            =&\frac12R_{\A\B ij}e_ie_j\psi_1-\frac12k_{\B i}e_i{\hat{\mathbf{n}}} k_{\A l}e_l{\hat{\mathbf{n}}} \psi_1+\frac12k_{\A i}e_i{\hat{\mathbf{n}}} k_{\B l}e_l{\hat{\mathbf{n}}}\psi_1\\
            &-k_{\B i}e_i h_{\A \C}e_\C \psi_1+k_{\A i}e_i h_{\B \C}e_\C \psi_1 -A_{\A\B i}e_i{\hat{\mathbf{n}}} \psi_1             \\
            =&\frac12\overline R_{\A\B ij}e_ie_j\psi_1-A_{\A\B i}e_i{\hat{\mathbf{n}}} \psi_1.
        \end{split}
    \end{align*}
Since $\overline R_{\A\B\C\D}=0$, we obtain in combination with \eqref{eq:R=-A}
    \begin{align}\label{123}
        0=\overline{R}_{\A\B\C{\hat{\mathbf{n}}}}e_\C{\hat{\mathbf{n}}}\psi_1-A_{\A\B\C}e_\C{\hat{\mathbf{n}}}\psi_1+A_{\A\B{\hat{\mathbf{n}}}}\psi_1
        =2\overline{R}_{\A\B\C{\hat{\mathbf{n}}}}e_\C{\hat{\mathbf{n}}}\psi_1+A_{\A\B{\hat{\mathbf{n}}}}\psi_1.
    \end{align}
 Recall the notation $(Z^\D)_\C=|\nabla u|^{-1}\operatorname{Im}\langle e_\C{\hat{\mathbf{n}}} \psi_1,\psi_1\rangle$ where $\psi_1$ is chosen such that $Z^\D $ asymptotes to $e_\D$.
 Multiplying \eqref{123} by $\psi_1$ and taking the imaginary part, yields
 \begin{align*}
     0=2|\nabla u|\overline{R}_{\A\B\C{\hat{\mathbf{n}}}}(Z^\D)_\C.
 \end{align*}
 Arguing as in the proof of Theorem \ref{T:spacetime harmonic} (4), this implies that $\overline{R}_{\A\B\C{\hat{\mathbf{n}}}}=0$.
 Moreover, equation \eqref{123} also gives $A_{\A\B{\hat{\mathbf{n}}}}=0$.
\end{proof}

\begin{lemma}
    We have
    \begin{align*}
      \overline   R_{\A{\hat{\mathbf{n}}}{\hat{\mathbf{n}}}\B}=A_{{\hat{\mathbf{n}}}\A\B}
    \end{align*}
\end{lemma}

\begin{proof}
We use the Hessian equation $\nabla^2u=-k|\nabla u|$
\begin{align*}
\begin{split}
    R_{\A {\hat{\mathbf{n}}}{\hat{\mathbf{n}}}\B}=&(\nabla_\A\nabla_{\hat{\mathbf{n}}}-\nabla_{\hat{\mathbf{n}}}\nabla_\A)\frac{\nabla_\B u}{|\nabla u|}\\
    =&-\nabla_\A k_{{\hat{\mathbf{n}}}\B}-\nabla_\A\left(\nabla_\B u\frac{k_{{\hat{\mathbf{n}}}{\hat{\mathbf{n}}}}}{|\nabla u|}\right)+\nabla_{\hat{\mathbf{n}}} k_{\A\B}+\nabla_{\hat{\mathbf{n}}}\left( \nabla_\B u\frac{k_{\A {\hat{\mathbf{n}}}}}{|\nabla u|}\right).
    \end{split}
\end{align*}
Applying the Hessian equation once more, the result follows.
\end{proof}

Hence, we may conclude that all Gauss and Codazzi terms are vanishing apart from $A_{{\hat{\mathbf{n}}}\A\B}$, also see Remark \ref{remark new}.
In particular, if $A_{{\hat{\mathbf{n}}}\A\B}=0$, $(M^n,g,k)$ must arise as spacelike slice of Minkowski space with second fundamental form $k$ by the fundamental theorem of hypersurfaces.
In the next section, we construct the Killing development and show it must be a pp-wave.

\section{The Killing development and proof of Theorem \ref{T:main}}\label{S:Killing}

Recall that $\{x_1,...,x_n\}$ are asymptotically flat coordinates and $u$ is a spacetime harmonic function asymptotic to $x_n$.
Moreover, the coordinates  $x_1,\cdots,x_n$ are extended into the interior to form $C^{s+1,\alpha}$-regular, globally defined functions.

In this section, we will introduce a new coordinate system $\{y_1,\dots,y_{n-1},u\}$ and  establish decay rate estimates for the coordinate function $y_A$ by elliptic estimates on the level-sets. To streamline the analysis, we first define notations for different types of asymptotic decay rates. %Here we use the coordinate system $\{x_1,...,x_{n-1},u\}$.
\begin{enumerate}
    \item Set $\pmb{\rho}:=\sqrt{x_1^2+\cdots+x_{n-1}^2}$ and $\rho:=\sqrt{y_1^2+\cdots+y_{n-1}^2}$, where $\{y_\A\}$ will be defined below. 
    \item The notion $\xi\in \mathcal{O}_{s,\alpha}(\pmb\rho^{-q})$ indicates the existence of a constant $C>0$ such that for all multi-indices $I\subset \{1,...,n-1\}$:
\[\|\xi\|_{C^{s,\alpha}_{-q}(\Sigma\setminus B)}:=
\sum_{|I|\le s}\left|\pmb{\rho}^{|I|+q}\partial_I \xi\right|+
\sum_{|I|=s}\sup_{\substack{\mathbf{x}_1,\mathbf{x}_2\in \R^{n-1}\setminus B
    \\ |\mathbf{x}_1-\mathbf{x}_2|\le \frac{|\mathbf{x}_1|}{2}}}\pmb{\rho}^{\alpha+s+q} \frac{|\partial_I \xi(\mathbf{x}_1,u)-\partial_I\xi({\mathbf{x}_2,u})|}{|\mathbf{x}_1-\mathbf{x}_2|^\alpha}\le C.
\]
\item Similarly, $\xi\in O_{s,\alpha}(\pmb\rho^{-q})$ indicates that there exists a constant $C$ such that for $I\subset \{1,...,n\}$, 
\[\|\xi\|_{\mathfrak{C}^{s,\alpha}_{-q}(M^n\setminus B)}:=
\sum_{|I|\le s}\left|\pmb{\rho}^{|I|+q}\partial_I \xi\right|+
\sum_{|I|=s}\sup_{\substack{x,y\in \R^n\setminus B 
    \\ |x-y|\le \frac{|x|}{2}}} 
    \rho^{\alpha+s+q}  \frac{|\partial_I \xi(x)-\partial_I\xi(y)|}{|x-y|^\alpha}\le C.
\]
\item Finally, $O_{s,\alpha}(r^{-q})$ is defined analogously to $O_{s,\alpha}(\pmb\rho^{-q})$ by replacing $\pmb\rho$ with $r$, and $\mathcal{O}_{s,\alpha}(\rho^{-q})$ is defined analogously to $\mathcal{O}_{s,\alpha}(\pmb{\rho}^{-q})$ by replacing $\pmb\rho$ with $\rho$.
\end{enumerate}
 We remark that $\rho$ and $\pmb{\rho}$ have the same growth rate. Therefore, after establishing some estimates for the coordinate system $\{y_1,\dots,y_{n-1},u\}$,  we will no longer distinguish between $\mathcal{O}_{s,\alpha}(\pmb{\rho}^{1-q})$ and  $\mathcal O_{s,\alpha}(\rho^{1-q})$. 
 
 For simplicity, we write $\|\xi\|_{C^{s,\alpha}_{-q}(\Sigma)}$ for the norm in (2) and $\|\xi\|_{\mathfrak{C}^{s,\alpha}_{-q}(M^n)}$ for the norm in (3), as we are only interested in the decay rates in the asymptotic region.

\subsection{Finding good coordinates}

On the one hand, we have
\begin{align*}
    g=|\nabla u|^{-2}du^2+g_\Sigma,
\end{align*}
where $g_\Sigma$ is the flat metric, and on the other hand, we have the asymptotically flat coordinate system $(x_1,\dots, x_n)$, cf. Definition \ref{Def:AF}. 
We will demonstrate that $g$ also  satisfies the following:

\begin{proposition}\label{Prop:good coordinates}
There exists a coordinate system $\{y_1,\cdots ,y_{n-1},u\}$ and a vector field $Y$ tangential to $\Sigma$ such that
\begin{equation}\label{eq:g0}
    g=(|\nabla u|^{-2}+|Y|^2)du^2+2\sum_{\A=1}^{n-1}Y_\A dudy_\A+\sum_{\A=1}^{n-1}dy_\A^{2}  .
\end{equation}
Moreover,   $(\nabla_{\hat{\mathbf{n}}})^{m} Y\in C^{s-m,\alpha}_{-m-q}(\Sigma)$ for $0\le m\le s-2$, $Y_\A\in \mathfrak{C}^{s-2}_{-q}(M^n)$, and $|\nabla u|^{-1}-1\in C^{s,\alpha}_{-q}(M^n)$.
\end{proposition}

This is already very close to the expression for the metric in  Theorem \ref{thm pp waves ids identities}.
However, at this point, we do not know yet that $Y$ is integrable on $\Sigma$ and that $Y=\nabla^\Sigma\ell$ for the graph function $\ell$.
Note that $Y_\A$ has a priori worse regularity and anistropic decay rates compared to $(|\nabla u|^{-1}-1)\in C^{s,\alpha}_{-q}(M^n)$. This discrepancy requires a more delicate analysis to achieve the optimal result.
In particular, we will show in Lemma \ref{ell decay} that $Y$ only loses a single order of regularity in the normal direction.

Recall that we assumed that the momentum vector $P$ is non-vanishing and $P=(0,\dots,0,-|P|)$. 
Moreover, we have shown in Theorem \ref{T:spacetime harmonic}(1) that the spacetime harmonic function $u$ satisfies near infinity
\begin{align*}
    u=x_n+O_{s+1,\alpha}(r^{1-q}).
\end{align*}
Therefore, $(x_1,\dots,x_{n-1},u)$ also forms an asymptotically flat coordinate system with decay rate $q$. 
On each level set $\Sigma$, using the coordinates $\mathbf{x}:=(x_1,\cdots, x_{n-1})$, we have $(g_\Sigma)_{\A\B}=\delta_{\A\B}+ O_{s,\alpha}(r^{-q})$, and $\det g_{\Sigma}=1+O_{s,\alpha}(r^{-q})$.
Moreover,
\begin{equation*}
    \Delta_{\Sigma} x_\A=\frac{1}{\sqrt{\operatorname{det}g_\Sigma}}\partial_\B(g^{\B\C}_{\Sigma} \sqrt{\operatorname{det}g_\Sigma} \partial_\C x_\A)=\mathcal O_{s-1}(r^{-q-1}).
\end{equation*}
Next, we construct a flat coordinate system $\{y_\A\}$ on the level-sets $\Sigma$ by solving harmonic equations.

   \begin{lemma}
      For $\A=1,2,\dots,n-1$,  there exists $y_\A$ such that
      \begin{equation} \label{eq:yalpha}
          \Delta_{\Sigma}y_\A=0, \quad y_\A=x_\A+\mathcal{O}_{s+1,\alpha}(\pmb{\rho}^{1-q}).
      \end{equation}
       Moreover, $\{y_\A\}$ forms a flat coordinate system for $g_\Sigma$, i.e., $g_\Sigma=dy_1^2+\cdots+dy_{n-1}^2$.
   \end{lemma} 
   
\begin{proof}
Since $u=x_n+\mathcal{O}_{s+1,\alpha}(r^{1-q})$, the coordinates $\{x_1,\dots,x_{n-1},u\}$ form a $C^{s,\alpha}_{-q}$ asymptotically flat system for $g$. Note that $\partial_{x_\A}$ are tangential to $\Sigma$ in this coordinate system.  

Consider the function $\xi_{0,R}$ defined by the elliptic equation 
\begin{equation*}  
    \Delta_\Sigma \xi_{0,R} = -\Delta_{\Sigma}x_\A, \quad \xi_{0,R} = 0 \ \text{on} \ |\mathbf{x}|=R.  
\end{equation*}  
Since $\Delta_\Sigma x_\A = \mathcal{O}_{s-1,\alpha}(r^{-1-q})$, there exists a constant $\widetilde{C}_0$ (independent of $\Sigma$) such that $\|\xi_{0,R}\|_{C_{1-q}^{s+1,\alpha}(\Sigma)} \leq \widetilde{C}_0$. Consequently, $\xi_{0,R} \to \xi_0$ subsequentially in $C^{s+1,\beta}(\Sigma)$ for any $\beta \in (0,\alpha)$. The limit $\xi_0$ satisfies  
\[  
\Delta_\Sigma \xi_0 = -\Delta_\Sigma x_\A, \quad \xi_0 \to 0 \ \text{at} \ \infty.  
\]  
More precisely, applying elliptic estimates and using  $\Delta_\Sigma x_\A = \mathcal{O}_{s-1,\alpha}(r^{-1-q})$, we conclude $\xi_0 =\mathcal{O}_{s+1,\alpha}(\pmb{\rho}^{1-q})$. 
Now let $y_\A = \xi_0 + x_\A$. Then $y_\A$ satisfies Equation \eqref{eq:yalpha}.  

Next, we restrict the discussion to a fixed level set $\Sigma$.  
Let $\{\tilde{y}_\A\}$ be a flat coordinate system for $g_{\Sigma}$.  
From the decay estimates of $y_\A$, we know that $|\nabla^\Sigma y_\A|$ is bounded.  
Therefore, $\partial_{\tilde{y}_\B} y_\A$ is a bounded harmonic function and thus constant by Liouville’s theorem.
Hence, $y_\A$ is a linear combination of $\{\tilde{y}_\B\}$ and is also a constant function.  
Using the identity  
\[
\langle \nabla^\Sigma y_\A, \nabla^\Sigma y_\B \rangle = \langle \nabla^\Sigma x_\A, \nabla^\Sigma x_\B \rangle + \mathcal{O}(\pmb{\rho}^{-q}) = \delta_{\A\B} + \mathcal{O}(\pmb{\rho}^{-q}),
\]  
we conclude $\langle \nabla^\Sigma y_\A, \nabla^\Sigma y_\B \rangle = \delta_{\A\B}$.  
\end{proof}

Next, we establish the decay rate estimates for the derivatives of $y_\A-x_\A$ along the $u-$direction. 
We remark that the proof below becomes significantly shorter for $s=2$.

\begin{lemma} \label{yalpha}
Let $y_\A$ be defined as above. Then we have for $0\le m\le s-1$
\begin{equation}\label{eq:y-x}
\partial_u^m(y_\A-x_\A) = \mathcal{O}_{s+1-m,\alpha}(\pmb{\rho}^{1-m-q}),  .
\end{equation}
Additionally, for any $\beta\in(0,\alpha)$,
\begin{equation} \label{eq:y-x beta}
    y_\A-x_\A= O_{s-1,\beta}(\pmb{\rho}^{1-q}), \;\; \;\; \nabla^\Sigma (y_\A-x_\A)= O_{s-1,\beta}(\pmb{\rho}^{-q}), \;\; \;\; (\nabla^\Sigma)^2 (y_\A-x_\A)= O_{s-1,\beta}(\pmb{\rho}^{-1-q}).
\end{equation}
\end{lemma}

\begin{proof} 
Throughout the proof, let $\widetilde{C}_i$, $\hat{C}_i$, $\overline{C}_s$ be constants which only depend on the initial data set and not on the individual slice $\Sigma$.

Denote with $\tilde{\Gamma}_{\A\B}^\C$ the  Christoffel symbols of $g_{\Sigma}$ under the coordinate system $\{x_1,...,x_{n-1}\}$. 

The Laplacian of the function $\xi_{0,R}$ defined above can be expressed as
\[
-\Delta_{\Sigma}x_\A=\Delta_\Sigma \xi_{0,R}=g_{\Sigma}^{\A\B}(\partial_{x_\A x_\B}\xi_{0,R}-\tilde{\Gamma}_{\A\B}^\C\partial_{x_\C}\xi_{0,R}),
\]
where we have the decay estimates
\[g^{\A\B}_\Sigma=\delta_{\A\B}+O_{s,\alpha}(r^{-q}),\quad
\tilde{\Gamma}_{\A\B}^\C=O_{s-1,\alpha}(r^{-q-1}),\quad\text{and } \Delta_{\Sigma}x_\A=O_{s-1,\alpha}(r^{-1-q}).\]
Observe that the coefficients of this elliptic equation are $C^{s-1}$-differentiable in the $u$-direction, which implies $\xi_{0,R}$  inherits the same $C^{s-1}$-differentiability. This regularity is verified by analyzing difference quotients in $u$, which is further used to establish the H\"older regularity in the $u$-direction, as shown in Equation \eqref{holder}.

Proceeding by induction on $m$, define $\xi_{m,R} = \partial_u^m\xi_{0,R}$ and assume for $0 \le m < s-1$ that there exist constants $\widetilde{C}_i$ (depending only on $(M^n,g,k)$) satisfying  
\[
   \|\xi_{i,R}\|_{C_{1-i-q}^{s+1-i,\alpha}(\Sigma)} \le \widetilde{C}_i \quad \text{for } 0 \le i \le m. 
\]

The base case $m=0$ has been already established above, and we proceed now with the induction step for $m+1$.  

Define $\xi_{m,R} := \partial_u^m \xi_{0,R}$ and $f_{m,R} := \Delta_\Sigma \xi_{m,R}$.  
Differentiating $f_{m,R}$ in $u$ yields 
\begin{align*}
    \begin{split}
        \partial_u f_{m,R} &= \partial_u (\Delta_\Sigma \xi_{m,R}) \\
        &= \partial_u\left[g_{\Sigma}^{\A\B}\left( \partial_{x_\A x_\B} \xi_{m,R} - \tilde{\Gamma}^\C_{\A\B}\partial_{x_\C} \xi_{m,R} \right)\right] \\
        &= \Delta_\Sigma(\partial_u \xi_{m,R}) 
           + \underbrace{(\partial_u g_{\Sigma}^{\A\B})\left( \partial_{x_\A x_\B}\xi_{m,R} - \tilde{\Gamma}^\C_{\A\B}\partial_{x_\C} \xi_{m,R} \right) 
           - g_\Sigma^{\A\B}(\partial_u \tilde{\Gamma}_{\A\B}^\C ) \partial_{x_\C}\xi_{m,R}}_{:=\Xi_{m,R}}  
    \end{split}
\end{align*}  
which implies $f_{m+1,R} = \partial_u f_{m,R} - \Xi_{m,R}$.
Recursive application of this relation yields  
\begin{equation}
f_{m+1,R} = \partial_u^{m+1}f_{0,R} - \sum_{i=0}^{m}\partial_u^{m-i}\Xi_{i,R}.
\end{equation}  

Using the decay estimates   
\[
g_{\Sigma}^{\A\B} = \delta_{\A\B} + O_{s,\alpha}(r^{-q}), \quad 
\tilde{\Gamma}^\C_{\A\B} = O_{s-1,\alpha}(r^{-1-q}), \quad 
\|\xi_{i,R}\|_{C_{1-i-q}^{s+1-i,\alpha}(\Sigma)} \le \widetilde{C}_i \quad (0 \le i \le m),
\]  
and the definition of $\Xi_{i,R}$, we obtain
\[
\|\partial_u^{m-i}\Xi_{i,R}\|_{C_{-2-m-2q}^{s-1-m,\alpha}(\Sigma)} \le \hat{C}_i
\]  
where $\hat{C}_i$ depends only on $(M^n,g,k)$.  
Combining this with $\partial_u^{m+1}f_{0,R} = -\partial_u^{m+1}\Delta_\Sigma x_\A$ and $\Delta_\Sigma x_\A = O_{s-1,\alpha}(r^{-1-q})$, we obtain  
\[
\|f_{m+1,R}\|_{C_{-m-2-q}^{s-m-2,\alpha}(\Sigma)} \le \overline{C}_{m+1}.
\]  
Applying elliptic estimates on the equation $\Delta \xi_{m+1,R}=f_{m+1,R}$ with Dirichlet condition yields
\[
\|\xi_{m+1,R}\|_{C_{-m-q}^{s-m,\alpha}(\Sigma)} \le \widetilde{C}_{m+1}.
\]  
This completes the inductive argument for all $0 \le m \le s-1$.

Next, we apply elliptic estimates to the Hölder coefficients of $\xi_{s-1,R}$.  
Specifically, for a small parameter $\epsilon > 0$, define the difference quotient  
\begin{align*}
    \zeta_R(\mathbf{x}, u, \epsilon) = \epsilon^{-\alpha} \left( \xi_{s-1,R}(\mathbf{x}, u+\epsilon) - \xi_{s-1,R}(\mathbf{x}, u) \right).
\end{align*}  
Given the inductive bounds $\|\xi_{i,R}\|_{C_{1-i-q}^{s+1-i,\alpha}(\Sigma)} \le \widetilde{C}_i$ for $0 \le i \le s-1$, we deduce 
$
\|f_{s-1,R}\|_{\mathfrak{C}^{0,\alpha}_{-s-q}(M^n)} \le \overline{C}_{s-1}$. 
Next, we decompose the Laplacian of $\zeta_R$  
\begin{equation} \label{holder}
    \begin{split}
        \Delta_{\Sigma_u} \zeta_R(\mathbf{x}, u, \epsilon) 
        &= \epsilon^{-\alpha} \left( \Delta_{\Sigma_{u+\epsilon}} \xi_{s-1,R}(\mathbf{x}, u+\epsilon) - \Delta_{\Sigma_u} \xi_{s-1,R}(\mathbf{x}, u+\epsilon) \right) \\
        &\quad - \epsilon^{-\alpha} \left( \Delta_{\Sigma_{u+\epsilon}} - \Delta_{\Sigma_u} \right) \xi_{s-1,R}(\mathbf{x}, u+\epsilon) \\
        &= \epsilon^{-\alpha} \left( f_{s-1,R}(\mathbf{x}, u+\epsilon) - f_{s-1,R}(\mathbf{x}, u) \right) + O\left( \epsilon^{1-\alpha} \pmb{\rho}^{-1-s-2q} \right),
    \end{split}
\end{equation}  
where we used the estimate  
\begin{equation*}
    \begin{split}
        & \left| \epsilon^{-\alpha} \left( \Delta_{\Sigma_{u+\epsilon}} - \Delta_{\Sigma_u} \right) \xi_{s-1,R}(\mathbf{x}, u+\epsilon) \right| \\
        &\le \epsilon^{-\alpha} \left| \left( g_{\Sigma_{u+\epsilon}}^{\B\C} - g_{\Sigma_u}^{\B\C} \right) \partial_{\B\C} \xi_{s-1,R}(\mathbf{x}, u+\epsilon) \right| \\
        &\quad + \epsilon^{-\alpha} \left| \left( g^{\B\C}_{\Sigma_{u+\epsilon}} \tilde{\Gamma}^{\A}_{\B\C}(\mathbf{x}, u+\epsilon) - g^{\B\C}_{\Sigma_u} \tilde{\Gamma}^{\A}_{\B\C}(\mathbf{x}, u) \right) \partial_\A \xi_{s-1,R}(\mathbf{x}, u+\epsilon) \right| \\
        &\le \hat{C}_{s+1} \epsilon^{1-\alpha} \pmb{\rho}^{-1-s-2q}.
    \end{split}
\end{equation*}  

This establishes the uniform bound  
\[
|\Delta_{\Sigma_u} \zeta_R(\mathbf{x}, u, \epsilon)| \le \overline{C}_s \pmb{\rho}^{-s-q-\alpha}.
\]  
Consequently, for any $\gamma \in (0,\alpha)$, the scaled term $\Delta_{\Sigma_u} \epsilon^{\alpha-\gamma} \zeta_R(\mathbf{x}, u, \epsilon)$ converges uniformly to $0$ as $\epsilon \to 0$.  
This implies uniform boundedness of $\|\xi_{s-1,R}(x,u)\|_{\mathfrak{C}^{0,\gamma}_{2-s-q}(M^n)}$ for all $\gamma \in (0,\alpha)$.  
Combining this with the derivatives estimates in tangential directions, we have 
\[
\|\xi_{0,R}(x,u)\|_{\mathfrak{C}^{s-1,\gamma}_{1-q}(M^n)} \le C \quad \text{(uniformly in $R$)}.
\]  
By the Sobolev embedding theorem, a subsequence $\xi_{0,R}$ converges to $\xi_0$ in $C^{s,\beta}$ for any $\beta \in (0,\gamma)$.   
The limiting function $\xi_0$ satisfies $\|\xi_0\|_{\mathfrak{C}^{s-1,\beta}_{1-q}(M^n)} \le C$. 
Moreover, applying elliptic estimates on  $\partial_u^m\xi_0$ for $0\le m\le s-1$, yields Equation \eqref{eq:y-x}, and we may choose any $\beta\in(0,\alpha)$, thereby completing the proof. 
\end{proof}

The above lemma implies that $\{y_1,\dots, y_{n-1},u\}$ forms a coordinate system such that
\begin{equation} \label{decay of Y}
    g=du^2+dy_1^2+\dots+ dy_{n-1}^2+O_{s-2,\beta}(\pmb{\rho}^{-q}).
\end{equation}
In these coordinates, the metric $g$ takes a particularly nice form.

\begin{proof}[Proof of Proposition \ref{Prop:good coordinates}]
We know that on $\Sigma$, the metric $g$ can be written as
$g|_{T\Sigma \otimes T\Sigma}=g_\Sigma=\sum_{\A=1}^{n-1}dy_\A^2$.
Moreover, we set $Y_\A=g(\partial_u, \partial {y_\A})$, and $f^2=g(\partial_u,\partial_u)-|Y|^2$,  {where $\partial_u=\frac{\partial}{\partial u}$ is the coordinate vector field}.
Hence,
\begin{equation*}
    g=(f^2+|Y|^2)du^2+2\sum_{\A=1}^{n-1}Y_\A dudy_\A+\sum_{\A=1}^{n-1}dy_\A^{2}.
\end{equation*}
Since $g^{-1}(du,du)=|\nabla u|^2$, and using the formula
\begin{equation}\label{eq g inverse}
    g^{-1}=\begin{bmatrix}
        f^{-2}& -f^{-2}Y^T
        \\ -f^{-2}Y & I_{n-1}+YY^{T}
    \end{bmatrix},
\end{equation}
we deduce that in fact $f=|\nabla u|^{-1}$. 

Since $Y_\A=-f^{2}g^{-1}(du,dy_\A)$, 
the regularity of $Y$ follows from Lemma \ref{yalpha} and we have
\begin{equation*}
    (\nabla_{\hat{\mathbf{n}}})^{m} Y\in C^{s-m,\alpha}_{-m-q}(\Sigma)
\end{equation*}
for $0\le m\le s-2$. Moreover, $Y_\A\in \mathfrak{C}^{s-2}_{-q}(M^n)$, and $(f-1)\in C^{s,\alpha}_{-q}(M^n)$.
\end{proof}

In the next section we will further analyze $Y$ and we will see that $Y$ is integrable on each level-set $\Sigma$.

\subsection{Applications of the Codazzi equations}

We begin this section with a technical lemma.

\begin{lemma}\label{k-identity lemma}
   Let $(M^n,g,k)$ be an initial data set where $g$ has the form from \eqref{eq:g0}, and $k=-|\nabla u|^{-1}\nabla^2u$.
Then we have
    \begin{equation} \label{kab}
    k_{\A\B}=\frac{1}{2}|\nabla u|(Y_{\A,\B}+Y_{\B,\A}), \quad\quad k_{{\hat{\mathbf{n}}} \A}=-|\nabla u|^{-1}\nabla_{\A}|\nabla u|,\quad\quad \text{and}\quad k_{{\hat{\mathbf{n}}}{\hat{\mathbf{n}}}}=-|\nabla u|^{-1}\nabla_{{\hat{\mathbf{n}}}}|\nabla u| 
\end{equation}
where $f=|\nabla u|^{-1}$.
\end{lemma}

\begin{proof} According to Equation \eqref{eq:g0}, 
 we have ${\hat{\mathbf{n}}}=|\nabla u|(\partial_u-Y_\A \partial y_\A)$, and
\begin{equation} \label{Yalphabeta}
   \Gamma_{\A\B}^u=\frac{1}{2}g^{uu}(g_{\A u,\B}+g_{\B u,\A})=\frac{1}{2}|\nabla u|^2(Y_{\A,\B}+Y_{\B,\A}).
\end{equation}
Here ``$,\A$'' denotes taking the derivative with respect to $\partial y_\A$.
Hence, we may compute 
\begin{equation*}
\begin{split}
    \nabla_{\A \B}u=&-\langle\nabla_\A \partial_\B,  {\hat{\mathbf{n}}}\rangle {\hat{\mathbf{n}}}(u)
    \\=&-|\nabla u|\langle{\hat{\mathbf{n}}} ,\Gamma_{\A\B}^u \partial_u\rangle= -\Gamma_{\A\B}^u
    \\=&-\frac{1}{2}|\nabla u|^{2}(Y_{\A,\B}+Y_{\B,\A}).
\end{split}
\end{equation*}
Therefore, the result follows from the identities $\nabla^2 u=-|\nabla u|k$.
\end{proof}

\begin{lemma} \label{ell decay}
Let $(M^n,g,k)$ be an initial data set where $g$ has the form from \eqref{eq:g0}, and  $k=-|\nabla u|^{-1}\nabla^2u$.
Moreover, suppose that the Codazzi equations $A_{\A\B\C}=0$ hold, cf. \ref{Lemma:Codazzi} item (ii).
Then the $1$-form $Y^\#$ dual to $Y$ with respect to $g_\Sigma$ is closed, i.e.,
    \begin{align*}
        d^\Sigma Y^\#=0.
    \end{align*}
    In particular, on each level set $\Sigma$ there exists a function $\ell$ such that 
\begin{equation*}
    Y=\nabla^{\Sigma}\ell.
\end{equation*}
Furthermore, when $n\ge 4$, for any $\beta\in(0,\alpha)$, we have $\ell \in \mathfrak{C}^{s-1,\beta}_{1-q}(M^n)$ and $(\nabla_{\hat{\mathbf{n}}})^m\ell \in C^{s+1-m,\beta}_{1-m-q}(\Sigma)$, where $0\le m\le s-1$; when $n=3$, $\ell_u$ and $|\nabla^\Sigma \ell|$ are bounded.
\end{lemma}

The function $\ell$ will be the graph function inside the pp-wave spacetime.
We remark that this regularity for $\ell$ also implies stronger regularity for $Y$ than the one initially obtained in Proposition \ref{Prop:good coordinates}.

\begin{proof}
First, note that ${\hat{\mathbf{n}}}=|\nabla u|(\partial_u-Y_\A\partial_\A)$.
Moreover, Equation \eqref{Yalphabeta} implies
\begin{equation*}
    \begin{split}
        \langle\nabla_{\A}\partial_\B,{\hat{\mathbf{n}}}\rangle=&\frac{1}{2}|\nabla u|(Y_{\A,\B}+Y_{\B,\A}),
        \end{split}
        \end{equation*}
        as well as
        \begin{equation*}
        \begin{split}
        \\ \langle\nabla_{\A}\partial_\B,\partial_\C\rangle=&g_{\C\C}\Gamma_{\A\B}^\C+g_{\C u}\Gamma_{\A\B}^u
        \\ =&\frac{1}{2}g^{\C u}(g_{\A u,\B}+g_{\B u, \A})+\frac{1}{2}Y_\C f^{-2}(g_{\A u,\B}+g_{\B u, \A})
=0.
    \end{split}
\end{equation*}
Hence, combining with Lemma \ref{k-identity lemma}, we have by the Codazzi equations $A_{\A\B\C}=0$ 
\begin{equation} \label{abck}
    \begin{split}
       0= & \nabla_\B k_{\A\C}-\nabla_{\A}k_{\B\C}
        \\=& \partial_\B k_{\A\C}-
        \langle\nabla_\B \partial_\A, {\hat{\mathbf{n}}}\rangle k_{{\hat{\mathbf{n}}} \C}-
        \langle\nabla_\B \partial_\C, {\hat{\mathbf{n}}} \rangle k_{\A{\hat{\mathbf{n}}}} -\partial_\A k_{\B\C}+ \langle\nabla_\A \partial_\B, {\hat{\mathbf{n}}}\rangle k_{{\hat{\mathbf{n}}} \C}+ 
        \langle\nabla_\A \partial_\C, {\hat{\mathbf{n}}}\rangle k_{\B{\hat{\mathbf{n}}}}
        \\=& \frac{1}{2}\partial_\B(|\nabla u|Y_{\A, \C}+|\nabla u|Y_{\C,\A}) 
        +\frac{1}{2}|\nabla u|(Y_{\B,\C}+Y_{\C,\B})\cdot |\nabla u|^{-1}\nabla_\A |\nabla u|
 \\&   -\frac{1}{2}\partial_\A (|\nabla u|Y_{\B,\C}+|\nabla u|Y_{\C,\B})-
  \frac{1}{2}|\nabla u|(Y_{\A,\C}+Y_{\C,\A})\cdot |\nabla u|^{-1}\nabla_\B |\nabla u|
  \\=& \frac{1}{2}|\nabla u|(Y_{\A,\B \C}-Y_{\B,\A\C}).
    \end{split}
\end{equation}
Thus, the term $(d^\Sigma Y^\#)_{\A\B}=Y_{\A,\B}-Y_{\B,\A}$ only depends on $u$ and is constant on each level set $\Sigma$.
Using Proposition \ref{Prop:good coordinates}, the term $(d^\Sigma Y)_{\A\B}$ must decay to zero on each level set $\Sigma$.   
Therefore, we can integrate $Y$ on $\Sigma$ to construct $\ell$. Moreover, applying the decay rate estimates of $Y_\A$ from Proposition \ref{Prop:good coordinates}, we can choose $\ell\to 0$ at $\infty$ on $\Sigma$ for $n\ge4$.
For $n=3$, we instead fix an integral curve $\gamma$ of $X$ and prescribe $\ell(\gamma)=0$.
Finally, the decay rate and regularity of $\ell$ can be obtained from the first equation of \eqref{kab}. 
More precisely, the second equation in Line \eqref{eq:y-x beta} implies that $\partial_\A=\nabla^\Sigma y_\A=\nabla^\Sigma x_\A+O_{s-1,\beta}(\rho^{-q})$.
Thus, using $k\in C^{s-1,\alpha}_{-1-q}(M^n)$, we have
\begin{equation} \label{eq:kAB}
    k_{\A\B}=|\nabla u|\nabla^\Sigma_{\A\B}\ell=O_{s-1,\beta}(\rho^{-1-q}) \quad \text{for any}\quad \beta\in(0,\alpha).
\end{equation}
Finally, integrating the equation above twice yields the decay estimates for $\ell$.
\end{proof}

Observe that the identity $ k_{\A\B}=|\nabla u|\nabla^\Sigma_{\A\B}\ell$ implies that the level-sets $\Sigma$ are umbilic in case $(M^n,g,k)$ arises as a $(t=-\ell=0)$-slice of a pp-wave spacetime.

\subsection{Defining the spacetime metric}
Recall that in the previous section we established 
\begin{equation}\label{eq:g}
    g=(|\nabla u|^{-2}+|\nabla^\Sigma\ell|^2)du^2+2\sum_{\A=1}^{n-1}\nabla^\Sigma_\alpha \ell dudy_\A+\sum_{\A=1}^{n-1}dy_\A^{2}.
\end{equation}
We now construct the Killing development, i.e. we define on $M\times \mathbb R=\mathbb R^{n+1}$ the Lorentzian metric
\begin{equation}\label{eq:bar g}
     \overline{g}=2d\tau du+g
\end{equation}
Next, we rewrite $\overline{g}$ in the typical pp-wave spacetime form, also see Theorem \ref{thm pp waves ids identities}.
{Let $\ell_u=\frac{\partial}{\partial u}\ell$.}
A short computation yields
\begin{equation}\label{eq:barg}
\begin{split}
     \overline{g}=& 2d(\tau +\ell)du+(|\nabla u|^{-2}+|\nabla_\Sigma\ell|^2-2\ell_u)du^2+\sum_{\A=1}^{n-1}dy_\A^2
     \\=& -2dt du+F(u,\mathbf{x})du^2+\sum_{\A=1}^{n-1}dy_\A^2
\end{split}
\end{equation}
where $t=-\tau-\ell$ and $F(u,\mathbf{x})=|\nabla u|^{-2}+|\nabla^\Sigma\ell|^2-2\ell_u$. 
Hence, $(M^n,g)$ is the graph $t=-\ell$ over the $(t=0)$-slice in $(\overline M^{n+1},\overline g)$. 
In the case $(M^n,g,k)$ is contained in Minkowski space, we obtain $F\equiv1$ and $t$ is just the usual time coordinate.

\begin{lemma}\label{Lemma:mu formula}
Consider the initial data set $(M^n,g,k)$ with $g$ given by \eqref{eq:g} and $k=-|\nabla u|^{-1}\nabla^2u$.
Then $(M^n,g)$ isometrically embeds into $(M^n\times \mathbb R,\overline g=2d\tau du+g)$ with second fundamental form $k$.
    Moreover,
    \begin{equation*}
    \mu=-\frac{1}{2}|\nabla u|^2\Delta_{\Sigma} F.
\end{equation*}
\end{lemma}

\begin{proof}
  Clearly, $(M^n,g)$ embeds isometrically into $(M^n\times\mathbb R,\overline g)$ as the constant $\tau$ slice.
Next, using the coordinate system $(\tau, u,y_1,\dots,y_{n-1})$,  we have
\begin{align*}
    e_0=|\nabla u|^{-1}(\partial_\tau-\nabla u).
\end{align*}
It is easy to see that $e_0$ is a timelike unit normal of $M^n\subset( \overline M^{n+1},\bar g)$.
Moreover, exploiting that $\partial_\tau$ is a covariantly constant vector field, we have
\begin{align*}
    \overline g(\overline{\nabla}_i e_0,e_j)=|\nabla u|^{-1}\overline g(\overline{\nabla}_i(\partial_\tau-\nabla u),e_j\rangle=k_{ij}.
\end{align*}
Hence, the second fundamental form of $M^n\subset( \overline M^{n+1},\bar g)$ equals $k$.
  Next, we compute
    \begin{equation*}
        \begin{split}
            \mu=&\frac{1}{2}(R-|k|^2+|\operatorname{tr}_gk|^2)
            \\=& \frac{1}{2}(R^{\Sigma}+2\operatorname{Ric}({\hat{\mathbf{n}}},{\hat{\mathbf{n}}})+|h|^2-H^2-|k|^2+|\operatorname{tr}_gk|^2).
        \end{split}
    \end{equation*}
    Recall the fact that $\nabla u\ne0$,
combining the identities $h=-k|_\Sigma$, $R^\Sigma=0$, and $\nabla^2 u=-|\nabla u|k$ with Bochner's formula 
 \begin{align*}
     |\nabla u|^2\operatorname{Ric}({\hat{\mathbf{n}}},{\hat{\mathbf{n}}})=\frac{1}{2}\Delta |\nabla u|^2-\langle\nabla \Delta u,\nabla u\rangle-|\nabla ^2 u|^2,
 \end{align*}
we obtain 
 \begin{equation*}
     \begin{split}
     \mu=&|\nabla u|^{-2}\left(\frac{1}{2}\Delta |\nabla u|^{2}+\langle \nabla (|\nabla u|\operatorname{tr}_g(k)), \nabla u \rangle\right)-|k|^2-Hk_{{\hat{\mathbf{n}}}{\hat{\mathbf{n}}}}-\sum_{\A=1}^{n-1}|k_{{\hat{\mathbf{n}}} \A}|^2
     \\=& |\nabla u|^{-2}\Bigg{(}\frac{1}{2}\nabla_{{\hat{\mathbf{n}}}{\hat{\mathbf{n}}}} |\nabla u|^{2}+H|\nabla u|\nabla_{{\hat{\mathbf{n}}}} |\nabla u|+\frac{1}{2}\Delta_\Sigma |\nabla u|^2
     +\operatorname{tr}_g(k)|\nabla u|\nabla_{\hat{\mathbf{n}}} |\nabla u|
     \\&+|\nabla u|^{2}\nabla_{{\hat{\mathbf{n}}}}\operatorname{tr}_g(k)\Bigg{)}-|k|^2-Hk_{{\hat{\mathbf{n}}}{\hat{\mathbf{n}}}}-\sum_{\A=1}^{n-1}|k_{{\hat{\mathbf{n}}} \A}|^2
     \end{split}
     \end{equation*}
     where in the second equation we split up the Laplacian into its normal and tangential components.
     Using identity \eqref{kab}, we see that the above term equals
     \begin{equation*}
     \begin{split}
     =& |\nabla u|^{-2}\left(\frac{1}{2}\nabla_{{\hat{\mathbf{n}}}{\hat{\mathbf{n}}}} |\nabla u|^{2}-|\nabla_{\hat{\mathbf{n}}} |\nabla u||^2+\frac{1}{2}\Delta_\Sigma |\nabla u|^{2}+
     |\nabla u|^{2}\nabla_{{\hat{\mathbf{n}}}}(|\nabla u|Y_{\A,\A}-|\nabla u|^{-1}\nabla_{{\hat{\mathbf{n}}}}|\nabla u|)\right)
     \\ &- |\nabla u|^{-2}|\nabla_{{\hat{\mathbf{n}}}} |\nabla u||^2-3|\nabla u|^{-2}\sum_{\alpha}|\nabla_{\alpha}|\nabla u||^2-f^{-2}\sum_{\A,\B=1}^{n-1}|Y_{\A,\B}|^2- Y_{\A,\A}\nabla_{{\hat{\mathbf{n}}}}|\nabla u|
     \\ =& |\nabla u|^{-2}\left(|\nabla u|^{3}\nabla_{\nabla_{{\hat{\mathbf{n}}}}{\hat{\mathbf{n}}}}|\nabla u|^{-1}-\frac{1}{2}|\nabla u|^{4}\Delta_{\Sigma}|\nabla u|^{-2}
     +|\nabla u|^{3}\nabla_{{\hat{\mathbf{n}}}}Y_{\A,\A}
     \right)
     \\&-|\nabla u|^{2}\sum_{\A,\B=1}^{n-1}|Y_{\A,\B}|^2
     +|\nabla u|^{2}\sum_{\A=1}^{n-1}|\nabla_\A |\nabla u|^{-1}|^2
     \\=& |\nabla u|^2\left(-\frac{1}{2}\Delta_{\Sigma}|\nabla u|^{-2}+|\nabla u|^{-1}\nabla_{{\hat{\mathbf{n}}}}Y_{\A,\A}-\sum_{\A,\B=1}^{n-1}|Y_{\A,\B}|^2\right)
     \end{split}
 \end{equation*}
 where in the last step we exploited $\nabla_{{\hat{\mathbf{n}}}}{\hat{\mathbf{n}}}=(|\nabla u|^{-1}\nabla_{\A}|\nabla u|)\partial_\A$.
Since ${\hat{\mathbf{n}}}=|\nabla u|(\partial_{u}-Y_\A\partial_\A)$ and $Y=(\partial_{\A}\ell) \partial_{\A}$, we have
\begin{equation*}
    \begin{split}
        |\nabla u|^{-1}\nabla_{{\hat{\mathbf{n}}}}Y_{\A,\A}-\sum_{\A,\B=1}^{n-1}|Y_{\A,\B}|^2=& (\partial_u-Y_\A\partial_\A) \Delta_\Sigma \ell-\sum_{\A,\B=1}^{n-1}|Y_{\A,\B}|^2
        \\=& \Delta_\Sigma \ell_u-\frac{1}{2}\Delta_{\Sigma}|Y|^2.
    \end{split}
\end{equation*}
Hence, 
\begin{equation*}
    \mu=-\frac{1}{2}|\nabla u|^2\Delta_{\Sigma}(|\nabla u|^{-2}+|Y|^2-2\ell_u)
\end{equation*}
which finishes the proof. 
\end{proof}

\begin{remark}\label{remark new}
It turns out that also the identity
\begin{align*}
    \nabla^\Sigma_{\A\B}F=-2|\nabla u|^{-2}A_{{\hat{\mathbf{n}}}\A\B}
\end{align*}
holds. Note that its trace reduces to $|J|=-\frac{1}{2}|\nabla u|^2\Delta_{\Sigma} F$ as above.
Since this is not used elsewhere in the text, we will omit the proof.
\end{remark}

 \begin{proof}[Proof of Theorem \ref{T:main}]
     It only remains to show that on each level set $F$ is superharmonic (with respect to $g_\Sigma$).
     However, this follows immediately from the above lemma together with the assumption $\mu\ge|J|\ge0$.
 \end{proof}

\section{Analyzing the PDE and proof of Theorem \ref{Cor:main}}\label{S:PDE}

Having established Theorem \ref{T:main}, we proceed with the proof of Theorem \ref{Cor:main}.
By leveraging the superharmonicity of $F$  on the level-sets $\Sigma$ and incorporating the decay rate of $F$ given in the last equation of \eqref{eq: F decay}, Theorem \ref{Cor:main} clearly holds for $q> n-3$ and $n\ge 4$.   
In this section, we extend this to $q>n-1-s-\alpha$ by a more careful analysis. 
As demonstrated by Example \ref{Example} below, this is the optimal result when $s=2$. 
We expect that it is possible to construct similar examples for other values of $s$.
Also, we remark that the majority of the technical difficulties in the proof below only arise for $s\ge3$.

\begin{proof}[Proof of Theorem \ref{Cor:main}]
    For $E=0$ and $q>\tfrac{n-2}2$, we can use Theorem \ref{T:spacetime harmonic} to obtain spacetime harmonic functions which asymptote to any coordinate function at $\infty$.
    Using the computations in the previous section, we obtain that all Gauss and Codazzi terms $\bar{R}_{ijkl}$ and $A_{ijk}$ vanish.  
    Hence $(M^n,g,k)$ isometrically embeds into Minkowski space by the fundamental theorem of hypersurfaces.
    Moreover, as mentioned above, Theorem \ref{Cor:main} holds for $q> n-3\ge 1$ which in particular covers $n=4$. For $n=3$, $F$ is bounded and superharmonic, hence $F$ is constant on each level set, yielding  Theorem \ref{Cor:main}.
    It remains to study the case where $E=|P|$, $n>4$, and $q>n-1-s-\alpha$.

Throughout the proof, we use the coordinate system $\{y_1,\cdots,y_{n-1},u\}$, in which the metric $g$ takes a convenient form, as given in Equation \eqref{eq:g}. Moreover, recall that $\mathbf y=(y_1,\dots,y_{n-1})$ are coordinates in $\R^{n-1}$ and $\rho=|\mathbf y|$.
    
Fix a constant $\beta\in(0,\alpha)$ such that  $q>n-1-s-\beta$.
Recall from Proposition \ref{Prop:good coordinates}, Lemma \ref{ell decay} and Equation \eqref{eq:barg}, that we obtain a function $F=f^{2}+|\nabla^\Sigma\ell|^2-2\ell_u$ on $M^n=\mathbb R^{n-1}\times\mathbb R$ satisfying
\begin{equation} \label{eq: F decay}
\begin{split}
    \Delta_\Sigma F\le 0, 
     &\quad \Delta_\Sigma F\in L^1(\R^n)
    \\
    \partial_u^m F\in C^{s-m,\beta}_{-m-q} (\Sigma) \text{\;for\;} 0\le m\le s-2 &\quad and\quad  F-1\in \mathfrak{C}^{s-2,\beta}_{-q}(M^n),
\end{split}
\end{equation}
where $\Delta_\Sigma F\in L^1(\R^n)$ follows from $\mu=-\frac{1}{2}|\nabla u|^2 \Delta_\Sigma F\in L^1(M^n)$.

We begin by considering the case in which 
$F$ is radially symmetric on 
$\Sigma$, and demonstrate at the end that the general case reduces to this case.
%Without loss of generality, we can assume $F$ is radially symmetric on $\Sigma$,
Thus, we obtain
\begin{equation*}
    0\ge \Delta_\Sigma F=F_{\rho\rho}+\frac{n-2}{\rho}F_\rho= \rho^{2-n}\partial_\rho\left(\rho^{n-2}F_\rho\right).
\end{equation*}
Hence, $\rho^{n-2}F_\rho$ is monotone decreasing in $\rho$. 
Since $\mu$ and hence $\Delta_\Sigma F\in L^1(\R^n)$, we know that $\Delta_{\Sigma_u} F$ is integrable on $\Sigma_u$ for almost all $u$. 
Thus, for almost all $u$
\begin{equation} \label{Fcu}
    \infty> -\int_{\Sigma_u}\Delta_{\Sigma_u} Fd\mathbf{y}= \lim_{\rho\to\infty} -\int_{S^{n-2}(\rho)} F_\rho d\rho 
    =-\omega_{n-2}\lim_{\rho\to\infty} \rho^{n-2}F_\rho\ge 0.  
\end{equation}
Therefore, $\rho^{n-2}F_\rho$ converges for almost all $u$. 
Therefore, we can write $F_\rho= \tilde{c}(u)\rho^{2-n}+o(\rho^{2-n})$,  where $\tilde{c}(u)\in L^1(\R)$ and $\tilde{c}(u)\le 0$.
Note that since $\mu\in C^{s-2,\alpha}$ with $s\ge2$, $\tilde c(u)$ is well-defined pointwise though might be infinite on a set of measure zero.
Integrating $F_\rho$ on $\Sigma$ with respect to $\rho$ and using the fact that $F\to 1$ at $\infty$, we find
\begin{equation}\label{eq:F-1=c}
F-1=c(u)\rho^{3-n}+o(\rho^{3-n}), \quad\text{where }\quad c(u)=\frac{\tilde{c}(u)}{3-n}.
\end{equation}
In case $c(u)$ inherits the same regularity as $(F-1)$, i.e. $c(u)\in  C^{s-2,\beta}(\R)$, and in case there are no lower-order terms present, we can just take $(s-2,\beta)$-derivatives of equation \eqref{eq:F-1=c} in $u$ direction to obtain the result.
In case $c(u)$ is less regular, we proceed by approximation.

If $(M^n,g,k)$ is not vacuum (i.e., not contained in Minkowski space), the function $c(u)$ is not everywhere vanishing and 
{{using $c(u)\in L^1(\R)$,}}
we can choose 
$u_{0} <u_{1}<\cdots <u_{s}$ such that for $0\le j\le s$,
\begin{equation} \label{u0j}
    c(u_{0})\ge 2^{sj}c(u_{j})\ge 0,\quad  |u_{j}-u_{j+1}|\le 2^{-j}|u_{j+1}-u_{j+2}|  \quad \text{and} \quad c(u_{0})>0.
\end{equation}

 When $s\ge 3$, we define for $i=1,2,\dots, s-1$,
\begin{equation} \label{Liu}
    \mathcal{L}_i(u)=\prod_{\substack{1\le j\le s-1\\ j\neq i}}\frac{u-u_j}{u_i-u_j}.
\end{equation}
Using Lagrange interpolation polynomials approximation, we have for $u\in [u_0,u_{s}]$
\begin{equation} \label{reminder}
    F(u,\rho)-1=\sum_{i=1}^{s-1}\mathcal{L}_i(u) [F(u_i,\rho)-1]+\mathcal{R}(u,\rho)
\end{equation}
where $\mathcal{R}(u,\rho)$ is the reminder term given by
\begin{equation} \label{Rurho}
    \mathcal{R}(u,\rho)= \frac{\prod_{i=1}^{s-1}(u-u_{i})}{(s-2)!}\partial_u^{s-2} F(\Phi(u,\rho),\rho)
\end{equation}
with $\Phi(u,\rho)\in[u_0,u_{s}]$.
Moreover, when $s=2$, we define $\mathcal{L}_1(u)=1$ and $\mathcal{R}(u,\rho)=F(u,\rho)-F(u_1,\rho)$.

{Applying \eqref{u0j}, we have} $|u_0-u_j|\le 2|u_{j-1}-u_{j}|\le |u_{j}-u_{j+1}|$ which implies that
\begin{equation} \label{Liu0}
    |\mathcal{L}_i(u_0)|=
    \prod_{\substack{1\le j\le s-1\\ j\neq i}}\frac{|u_0-u_j|}{|u_i-u_j|} \le 
    \prod_{i< j\le s-1 }\frac{u_j-u_0}{u_j-u_i}\le 2^{s-1-i}.
\end{equation}
Using Equation \eqref{eq:F-1=c}, the first and third inequality in \eqref{u0j}, we obtain for $\rho$ sufficiently large
\begin{equation}\label{u0jF}
\begin{split}
     |F(u_j,\rho)-1|\le& 2^{1-sj}c(u_0)\rho^{3-n}\quad \text{for} \quad j\ge 1,
     \\ F(u_0,\rho)-1\ge& \frac{9}{10} c(u_0)\rho^{3-n} .
\end{split}
\end{equation}
Hence, when $s=2$, for $\rho$ sufficiently large,
\begin{equation*}
    F(u_0,\rho)-F(u_1,\rho)\ge \frac{2}{5}c(u_0)\rho^{3-n}.
\end{equation*}
This contradicts the fact $(F-1)\in \mathfrak{C}^{0,\beta}_{-q}(M^n)$, where $q>n-1-s-\beta=n-3-\beta$, {{since $(F-1)\in \mathfrak{C}^{0,\beta}_{-q}(M^n)$ implies that  there exists a constant $C_{0,\beta}$ such that
\[|F(u_0,\rho)-F(u_1,\rho)|\le C_{0,\beta}|u_0-u_1|^{\beta}\rho^{-q-\beta}.\]
}} 

It remains to study the case $s\ge3$.
Using the inequalites in \eqref{u0jF}, \eqref{Liu0} and Equation \eqref{reminder}, we obtain 
\begin{equation*}
    \begin{split}
        \mathcal{R}(u_0,\rho)=&F(u_0,\rho)-1-\sum_{i=1}^{s-1}\mathcal{L}_i(u_0) [F(u_i,\rho)-1]
        \\ \ge& \frac{9}{10}c(u_0)\rho^{3-n}-\sum_{i=1}^{s-1}2^{s-1-i}  2^{1-si}c(u_0)\rho^{3-n}
        \\ \ge & 
        \left(\frac{9}{10}-\frac{1}{2-2^{-s}}\right) c(u_0)\rho^{3-n}
        \\ \ge & \left(\frac{9}{10}-\frac{4}{7}\right)c(u_0)\rho^{3-n}>\frac{1}{4} c(u_0)\rho^{3-n}.
    \end{split}
\end{equation*}
Consequently, {{with the help of Equation \eqref{Rurho}}}, we have for $s\ge3$
\begin{equation} \label{0xi}
    |\partial_u^{s-2} F(\Phi(u_0,\rho),\rho)|
    \ge \frac{1}{4} c(u_0)\rho^{3-n} \frac{(s-2)!}{\prod_{i=1}^{s-1}|u_0-u_{i}|}.
\end{equation}

To take the difference quotient of $\partial^{s-2}_u F$, {{we combine Equation \eqref{reminder} and \eqref{Liu}},} and plug $u_s$ into Equation \eqref{Rurho} to find
\begin{equation}
\begin{split}\label{sxi}
    |\partial_u^{s-2} F(\Phi(u_s,\rho),\rho)|=& \frac{(s-2)!}{\prod_{i=1}^{s-1}|u_s-u_{i}|} \cdot|\mathcal{R}(u_s,\rho)|
    \\ \le & \frac{(s-2)!}{\prod_{i=1}^{s-1}|u_s-u_{i}|} |F(u_s,\rho)-1|+\sum_{i=1}^{s-1}\frac{(s-2)!|F(u_i,\rho)-1|}{\displaystyle{\prod_{\substack{1\le j\le s \\ j\neq i}}|u_i-u_{j}|}}
    \\ \le& \frac{1}{8}c(u_0)\rho^{3-n}\cdot \frac{(s-2)!}{\prod_{i=1}^{s-1}|u_0-u_{i}|}
\end{split}
\end{equation}
where the last inequality is achieved by choosing
$u_s$ sufficiently large.
After taking the difference quotient, Equation \eqref{0xi} and \eqref{sxi} contradict the fact that $(F-1)\in \mathfrak{C}^{s-2,\beta}_{-q}(M^n)$ where $q>n-1-s-\beta$. 
Therefore, we must have $c(u)=0$.
Thus,  $F\equiv1$.

We now consider the general case where 
$F$ is not necessarily radially symmetric. Observe that the function $\rho^{2-n}\int_{S^{n-2}(\rho)} F$ is also a superharmonic function on $\Sigma$.
Moreover, the decay rates cannot deteriorate under this symmetrization. Applying the previous argument to the spherical average of $F$ on $S^{n-2}$,  we obtain 
\[\frac{1}{\rho^{n-2}\omega_{n-2}}\int_{S^{n-2}(\rho)}F=1.\]
Using the identity $\Delta_\Sigma=\partial_{\rho\rho}+\frac{n-2}{\rho}\partial_\rho+\frac{1}{\rho^2}\Delta_{S^{n-2}}$, we have 
\begin{equation*}
    0\ge\frac{1}{\rho^{n-2}}\int_{S^{n-2}(\rho)}\Delta_{\Sigma}F=
   \Delta_{\Sigma} \left(\frac{1}{\rho^{n-2}}\int_{S^{n-2}(\rho)}F\right)=0.
\end{equation*}
Then we obtain $\Delta_{\Sigma} F=0$. Combined with the fact $(F-1)\in \mathfrak{C}^{s-2,\beta}_{-q}(M^n)$ from line \eqref{eq: F decay}, we conclude that $F\equiv 1$. Therefore, the metric is vacuum, and $(M^n,g)$ can be embedded into the Minkowski spacetime.
\end{proof}

While Theorem \ref{Cor:main} rules out asymptotically flat pp-waves in dimensions 3 and 4, we point out that it is easy to construct \emph{incomplete} pp-waves in these scenarios.

\begin{remark}\label{rem:previous results}
    Theorem \ref{Cor:main} has already been established by P.~Chrusciel and D.~Maerten \cite{ChruscielMaerten} under the additional assumptions $q>n-3$, $\mu,J =\mathcal O(r^{-n-\varepsilon})$ for some $\varepsilon>0$, and $C^{3,\alpha}_{-q}$ asymptotical flatness.
Moreover, it has been shown L.H.~Huang and D.~Lee \cite{HuangLee,HuangLee2} under the additional assumptions
$q>n-2-\alpha$,
$\mu,J =\mathcal O(r^{-n-\varepsilon})$ for some $\varepsilon>0$,
and $g\in C^5$ as well as $k\in C^4$.
We note that instead of a spin assumption, L.H.~Huang and D.~Lee assume that the spacetime positive mass theorem holds in an admissible neighborhood of $(M^n,g,k)$ (which it automatically does in the spin case).
Finally, there is a proof by the authors \cite{HirschZhang} using spacetime harmonic functions which additionally assumes that $n=3$.

The spacetime positive mass theorem itself has been established in \cite{Eichmair, EHLS, HKK, SY2, Witten} using various techniques and we refer to \cite{HKK} for a more detailed overview.
\end{remark}

\section{Non-trivial pp-wave spacetimes and the proof of Theorem \ref{Thm:Example}}\label{S:example}

Recall that in Example \ref{Ex HL} by L.-H.~Huang and D.~Lee, the pp-wave metric is given by $\overline g=-2dudt+Fdu^2+g_{\mathbb R^{n-1}}$, where $F=1+\eta(u)\kappa(y_1,\dots,y_{n-1})$ for a positive cut-off function $\eta:\mathbb R\to[0,1]$ and a superharmonic function $\kappa:\mathbb R^{n-1}\to\mathbb R$.
The reason the $(t=0)$-slice of this spacetime is not $C^{2,\alpha}_{-q}$-asymptotically flat in low dimensions ($n\le8$) hinges upon the fact that the decay of $F$ (with respect to $\rho=\sqrt{y_1^2+\dots+y_{n-1}^2}$) does not improve when taking derivatives in $u$-direction.
More precisely, $\partial_uF=\eta' \kappa$ and $\partial_{uu}F=\eta''\kappa$ still have the same decay as $\kappa$ since the cutoff function $\eta$ - or rather its derivatives - do not decay as $\rho\to\infty$.

Instead of changing the cutoff function $\eta$ in Example \ref{Ex HL}, we can also consider certain graphs in this spacetime.
More precisely, for each positive constant $c$, we can construct a graph function $\ell(u,\rho)$ defined in Equation \eqref{Gellup} below and satisfies $\operatorname{supp}(\ell(\cdot,\rho))\subseteq (-\kappa^{-c}(\rho),\kappa^{-c}(\rho))$. 
Comparing the metric coefficient $G(u,\rho)$ defined in \eqref{Gellup} and $F=1+\kappa(\rho)\eta(u)$, 
we make the key observation is that  the derivatives of $G$ in the $u$-direction exhibit improved decay rates, which strengthen as $c$ increases.
However,  this improvement introduces a trade-off: the metric $g$ acquires an additional cross-term $2\ell_\rho d\rho du$, whose decay rates deteriorate with larger $c$.  
In the following example, we determine the correct $c$ and construct an initial data set with the sharp decay rate.
More precisely, the initial data set is $C^{2,\alpha}_q$-asymptotically flat for $q=n-3-\alpha$ which is optimal in view of Theorem \ref{Cor:main}.

\begin{example}\label{Example}
    Let $n\ge5$, {$\alpha\in(0,1)$}, and let $\eta:\mathbb R\to \mathbb R$ be a smooth, even, non-negative function compactly supported on $[-1,1]$ with a sufficiently small $C^{3}$ norm\footnote{This can always be achieved by replacing $\eta$ with $\epsilon \eta$, $\epsilon\ll1$.}.
    Let $\kappa:\mathbb R^{n-1}\to\mathbb R$ be a radially symmetric smooth superharmonic function which equals $\rho^{3-n}$ outside the compact set $\{\rho \ge 1\}$.
    Define on $M^n=\mathbb R^{n-1}\times\mathbb R$, the functions
    \begin{equation}\label{Gellup}
        \begin{split}
    G(u,\rho)=&1+2\kappa^{1+c}(\rho)\eta (\kappa^c(\rho)  u),
    \\ \ell(u,\rho)=&\int^u_{-\infty} \left(\kappa^{1+c}(\rho)\eta(\kappa^c(\rho) t)-\kappa(\rho) \eta(t)\right)dt,
    \end{split}
    \end{equation}
    where $c=\tfrac1{n-3}$. 
Then $(M^n,g,k)$ given by
        \begin{equation}
        \label{metric}
        g=Gdu^2+2\ell_\rho d\rho du+\sum_{\A=1}^{n-1}dy_\A^2
    \end{equation}
and
        \begin{equation*}
        k=-|\nabla u|^{-1}\nabla^2 u
    \end{equation*}
   is a $C^{2,\alpha}$-initial data set $(M^n,g,k)$ with decay rate $q=n-3-\alpha$.
   Moreover, $(M^n,g,k)$ satisfies the dominant energy condition, has $E=|P|>0$, and does not embed into the Minkowski spacetime.
\end{example}

For $(M^n,g,k)$ to be asymptotically flat, one needs $n\ge 5$ such that one can find an $\alpha\in(0,1)$ with $n-3-\alpha>\tfrac{n-2}2$.
In particular, there are no such examples for $n\le4$.
This adds to a list of other results which satisfy similar dimensional restrictions.
For instance, minimal hypersurfaces are smooth up to dimension 6, and spinors are pure up to dimension 6.

The proof below is quite subtle and requires a delicate asymptotic analysis. 
First, note that choosing $\|\eta\|_{C^{3}}$ sufficiently small, we have $G>|\ell_\rho|^2$ which implies that $g$ is positive definite.
We remark that the initial data set above embeds into the pp-wave spacetime from Example \ref{Ex HL}, and in particular, all the results from Section \ref{S:pp waves overview} apply.
Notably, we have $E=|P|\ne0$ by Theorem \ref{3 thm 2} and 
$$\mu=|J|= -\frac{1}{2}|\nabla u|^2\Delta_{\Sigma} (G-2\ell_u)=-\frac{1}{2}|\nabla u|^2\Delta_{\Sigma} [2\eta(u)\kappa(\rho)]=-\eta(u)|\nabla u|^2\Delta_{\Sigma}\kappa(\rho) $$
by Theorem \ref{3 thm 1} which implies $\mu\in L^1$. 
Hence, it suffices to show the decay estimates for $g$ and $k$.

We will first prove that $g$ is $C^{2,\alpha}_{4-\alpha-n}$ with respect to the $(u, y_1,\cdots y_{n-1})$ coordinate system.
Next, we will reparameterize the metric to demonstrate that $(M^n,g,k)$ has the optimal decay rate under this new coordinate system.

\begin{lemma}
Using the notation of Example \ref{Example}, we have
\begin{equation} \label{Glk}
  \text{(i)}\,\,  G-1\in C^{2,\alpha}_{2-n}(M^n), \quad\quad
     \text{(ii)}\,\, \ell_\rho\in C^{2,\alpha}_{4+\alpha-n}(M^n), \quad \quad
     \text{(iii)}\,\,  k\in C^{1,\alpha}_{2+\alpha-n}(M^n)
\end{equation}
with respect to the coordinate system $(u, y_1,\cdots, y_{n-1})$.
\end{lemma}

\begin{proof}
We begin with showing estimate (i).
Since $\eta$ is compactly supported on $[-1,1]$, and since $\kappa^c=\rho^{-1}$ for $\rho\ge 1$, we obtain that $G(u,\rho)\equiv 1$ on the set $\{|u|\ge \rho\ge 1\}$.
Therefore, we only need to verify the asymptotics in the region $\{|u|\le \rho\}$.
Note that in this region we have $G(u,\rho)=1+2 \rho^{2-n}\eta(\rho^{-1}u)$. 
Let $I$ be an index set consisting by the coordinates $\{u,y_1,\cdots y_{n-1}\}$. Observe that $|\partial_I[\eta(\rho^{-1} u)]|=O(r^{-|I|})$. This implies $G-1\in C^{i,\alpha}_{2-n}(M^n)$ for any $i\in \mathbb{N}$.

Next, we demonstrate that estimate (ii) holds.
Since $\eta(t)$ is compactly supported on $[-1,1]$ and 
\begin{equation*}
    \int_{-\infty}^\infty \kappa^{1+c}\eta(\kappa^c t)-\kappa \eta(t) dt=0,
\end{equation*}
the function $\ell(u,\rho)$ is  supported on $\{|u|\le \max\{1, \kappa^{-c}\}\}$. 
Therefore, it suffices to perform the computations in the asymptotic region of $\{|u|\le \rho\}$. % and $\kappa=\rho^{3-n}$.  
 Note that in this region we have $\ell_u= \rho^{2-n}\eta(\rho^{-1}u)-\rho^{3-n}\eta (u)$.
Hence, 
    $\ell_u=O_i(r^{2-n})-\rho^{3-n}\eta(u)$, for any $i\in \mathbb{N}$. Therefore, for any fixed $u$,  $(\partial_u)^m\ell\in C^{i,\alpha}_{3-n}(\mathbb{R}^{n-1})$ for any $i\in \mathbb{N}$. 
    In particular, $\ell_\rho\in C^{2,\alpha}_{4+\alpha-n}(M^n)$.

Finally, we verify that $k\in C^{1,\alpha}_{2+\alpha-n}(M^n)$. 
Note that $|\nabla u|^{-2}=G-|\ell_\rho|^2=1+O_{2,\alpha}(r^{2-n})$. 
Applying Equation \eqref{kab} and the decay estimates of $\ell$, we have
\begin{equation*}
    k_{{\hat{\mathbf{n}}} i}=-|\nabla u|^{-1}\nabla_i |\nabla u|=O_{1,\alpha}(r^{1-n}), \quad \quad\text{and}\quad\quad
    k_{\A\B}=|\nabla u|\ell_{\A\B}=O_{1,\alpha}(r^{2+\alpha-n})
\end{equation*}
which finishes the proof.
\end{proof}

From the estimates above, $g$ is only a $C^{2,\alpha}_{4+\alpha-n}$-asymptotically flat metric because of the cross term $2\ell_\rho dud\rho$ in the metric. Therefore, a reparameterization that absorbs this cross term will improve the decay rates.
Let us first study a function $L$ defining in the following lemma which will be used in the reparametrization. 

\begin{lemma} \label{A1}
Define 
\begin{equation} \label{Lrhouu}
L(u, \rho)=\int^u_{-\infty}\ell(t,\rho)dt.    
\end{equation}
Then $L(u,\rho)$  vanishes on $\{|u|\ge \kappa^{-c}(\rho)\}$. Moreover, for any $i\in\mathbb N$
\begin{equation} \label{Ldecay}
     L\in C^{i,\alpha}_{4-n}(\R^{n-1}) \quad \text{and} \quad  
     (\partial_u)^mL\in C^{i,\alpha}_{3-n}(\R^{n-1}) \quad\text{when}\quad m\ge 1.
\end{equation}
\end{lemma}
\begin{proof}
  As $\eta$ is an even function defined in Example \ref{Example}, we have
\begin{equation*}
    \begin{split}
        \int_{-\infty}^{\infty} \ell(u,\rho)du=&\int_{-\infty}^{\infty}\int^u_{-\infty} \left(\kappa^{1+c}(\rho)\eta(\kappa^c(\rho) t)-\kappa(\rho) \eta(t)\right)dtdu
        \\=&\kappa(\rho)\int_{-\infty}^{\infty}\left[
        \int^{\kappa^{c}(\rho)u}_{-\infty}\eta(t)dt-\int^u_{-\infty}\eta(t)dt
        \right]du
        \\=& \kappa(\rho)\int_{0}^{\infty}
        \left[\int^{\kappa^{c}(\rho)u}_{u}\eta(t)dt+\int^{-\kappa^{c}(\rho)u}_{-u}\eta(t)dt\right]du=0.
    \end{split}
\end{equation*}
Since $\ell(u,\rho)$ vanishes on $\{|u|\ge \kappa^{-c}(\rho)\}$, the function $L(u,\rho)$ also vanishes on $\{|u|\ge \kappa^{-c}(\rho)\}$.
Therefore, we only need to consider the asymptotic region in $\{|u|\le \kappa^{-c}(\rho)\}$. 
In this case, $\ell_u= \rho^{2-n}\eta(\rho^{-1}u)-\rho^{3-n}\eta (u)$.
Integrating $\ell_u$ twice, we obtain with the help of Equation \eqref{Lrhouu}
\begin{equation} \label{Lurho}
\begin{split}
    L(u,\rho)=&\int^u_{-\infty}\int_{-\infty}^w \left(\rho^{2-n}\eta(\rho^{-1}t)-\rho^{3-n}\eta (t)\right) dt dw
    \\=&\rho^{4-n}\int_{-\infty}^{\rho^{-1}u}\int_{-\infty}^\mathbf{w}\eta(t)dtd\mathbf{w}-\rho^{3-n}\int^u_{-\infty}\int^w_{-\infty}\eta(t)dtdw
    \\=&\rho^{4-n}\vartheta(\rho^{-1}u)-\rho^{3-n}\vartheta(u)
\end{split}
\end{equation}
where $\vartheta(u)=\int^u_{-\infty}\int^w_{-\infty}\eta(t)dtdw$. 
Since $L(u,\rho)$ vanishes on $\{|u|\ge \kappa^{-c}(\rho)=\rho\}$, we obtain 
\begin{equation*}
    \vartheta(\rho^{-1}u)=\rho^{-1}\vartheta(u)\quad \text{when }\quad |u|\ge \rho.
\end{equation*}
This means $\vartheta(u)$ is a linear function when $u\ge 1$. On the other hand, we have by definition $\vartheta=0$ when $u\le -1$.
Therefore, Equation \eqref{Lurho} implies
\begin{equation*}
     L\in C^{i,\alpha}_{4-n}(\R^{n-1}) \quad \text{and} \quad  
     (\partial_u)^mL\in C^{i,\alpha}_{3-n}(\R^{n-1}) \text{ when } m\ge 1
\end{equation*}
for any $i\in \mathbb N$.
\end{proof}

\begin{lemma} Let $\varrho=\rho+L_\rho$ and denote with $d\varsigma^2$ the round metric on $S^{n-2}$.
    With respect to the coordinates $(u,\varrho,\varsigma)$, the metric $g$ defined in equation \eqref{metric} is $C^{2,\alpha}_{3+\alpha-n}$-asymptotically flat. In particular, this proves Example \ref{Example}.
\end{lemma}

\begin{proof}
We have
\begin{equation*}
\begin{split}
    d\varrho^2=&[d(\rho+L_\rho)]^2
  \\=& [d\rho+L_{\rho\rho}d\rho+\ell_\rho du]^2
  \\=&d\rho^2+2\ell_\rho dud\rho+2L_{\rho\rho}d\rho^2+(dL_\rho)^2.
\end{split}
\end{equation*}
Therefore, the metric $g$ defined in equation \eqref{metric} takes the form
\begin{equation*}
\begin{split}
    g=&Gdu^2+d\varrho^2+\rho^2 d\varsigma^2-2L_{\rho\rho}d\rho^2-(dL_\rho)^2
    \\=&g_{\R^n}+(G-1)du^2+(\rho^2-\varrho^2)d\varsigma^2-2L_{\rho\rho}d\rho^2-(dL_\rho)^2.
\end{split}
\end{equation*}
To prove the $C^{2,\alpha}_{3+\alpha-n}$-asymptotic for $g$, it suffices to show
\begin{equation*}
    G-1,\;\; (\rho^2-\varrho^2)\varrho^{-2} \;\in C^{2,\alpha}_{3+\alpha-n}(M^n),\quad \quad\text{and}\quad
    \quad L_{\rho\rho}d\rho^2,\;\;(dL_\rho)^2
    \in C^{2,\alpha}_{3+\alpha-n}(M^n)
\end{equation*}
with respect to the coordinate system $(u,\varrho,\varsigma)$.
To distinguish these two coordinate systems, we write $\{z_1,\cdots ,z_n\}$ for $\{u,\rho,\varsigma\}$ and $\{\tilde{z}_1,\cdots ,\tilde{z}_n\}$ for $\{u,\varrho,\varsigma\}$.
For any $C^{2}$ function $\mathcal{F}$ we have
\begin{equation*}
\begin{split}
    \frac{\partial^2 \mathcal{F}}{\partial \tilde{z}_i \partial \tilde{z}_j}=&\frac{\partial}{\partial \tilde{z}_i} \left(\frac{\partial \mathcal{F}}{\partial z_p}\cdot\frac{\partial z_p}{\partial \tilde{z}_j}\right)
    \\= & \frac{\partial \mathcal{F}}{\partial z_p}\cdot\frac{\partial^2 z_p}{\partial \tilde{z}_i\partial \tilde{z}_j}+
    \frac{\partial^2 \mathcal{F}}{\partial z_q \partial z_p}
    \cdot \frac{\partial z_q}{\partial \tilde{z}_i}\cdot
    \frac{\partial z_p}{\partial \tilde{z}_j}.
\end{split}
\end{equation*}
Since
\begin{equation*}
    \frac{\partial(u,\varrho,\varsigma)}{\partial(u,\rho,\varsigma)}=
    \begin{bmatrix}
        1&0&0
        \\ L_{\rho u}&1+L_{\rho\rho}&0
        \\ 0&0&I_{n-2}
    \end{bmatrix}, 
\end{equation*}
we can use equation \eqref{Ldecay} to obtain
\begin{equation*}
    \frac{\partial z_p}{\partial \tilde{z}_i}=\delta_{ip}+O(r^{2-n}), \quad \text{and}\quad
    \frac{\partial^2 z_p}{\partial \tilde{z}_i\partial \tilde{z}_j}=O(r^{2-n}).
\end{equation*}
Moreover, the $C^{0,\alpha}$ H\"older norm of $\frac{\partial^2 z_p}{\partial \tilde{z}_i\partial \tilde{z}_j}$
is still $O(r^{2-n})$. 
Therefore, it remains to demonstrate that 
\begin{equation*}
    G-1,\;\; (\rho^2-\varrho^2)\varrho^{-2}, \;\; L_{\rho\rho},\;\;L_{\rho u}^2,\;\; L_{\rho\rho}^2, \;\;L_{\rho u} L_{\rho\rho}
    \in C^{2,\alpha}_{3+\alpha-n}(M^n)
\end{equation*}
with respect to the coordinate system $(u,\rho,\varsigma)$.
The decay rates of the last four terms follow from Lemma \ref{A1}.
Moreover, we have $G-1\in C^{2,\alpha}_{2-n}$ according to Equation \eqref{Glk}.  
Finally, we compute 
\begin{equation*}
    (\rho^2-\varrho^2)\varrho^{-2}
    =L_\rho^2\varrho^{-2}-
    2L_\rho\varrho^{-1}\in C^{2,\alpha}_{3+\alpha-n}(M^n),
\end{equation*}
which finishes the proof.
\end{proof}

\appendix

\section{Trivial topology of pp-waves}\label{S:topology}

We will show that each level set is topologically trivial (i.e., $\Sigma\cong\mathbb R^{n-1}$), which implies that $M^n\cong \R^n$. The main theorem \ref{trivial top} is a direct corollary of the Reeb's stability theorem with boundary, stated at the end of the appendix  (\Cref{Reeb thm with boundary}).
For completeness, we adapt the proof of the Reeb's stability theorem from \cite[p.72, Theorem 4]{Camacho-Neto} to our setting.

\begin{theorem}\label{trivial top}
    Let $M^n$ be an asymptotically flat manifold of decay rate $q>0$ with a single end. Suppose $X$ is a smooth globally defined nowhere vanishing vector field on $M^n$ with a closed dual $1$-form, and that $X=\nabla x_n+O_1(r^{-q})$ holds in the asymptotic region. Then there exists a global foliation $\mathfrak{F}$ such that each leaf $\mathfrak{L}$ is diffeomorphic to $\R^{n-1}$ with normal vector $X$. Moreover, $M^n$ is diffeomorphic to $\mathfrak{L}\times \R=\R^{n}$.
\end{theorem}

Since $X$ is closed and the asymptotic region is simply connected, we can construct a function $u$ in the asymptotically flat end such that $du=X$ and such that $u$ asymptotes to $x_n$.
By the global Frobenius theorem, there exists a global foliation $\mathfrak{F}$ on the entire manifold $M^n$. Define $\mathcal{B}=\{u=-C_0, x_1^2+\cdots+x_{n-1}^2\le C_0^2\}$ where $C_0\gg1$ is chosen so large such that $\mathcal B$ is contained entirely in the asymptotically flat end.
Let $\mathcal{C}$ be a large cylinder with bottom face $\mathcal{B}$,  side $\mathcal{S}$ tangent to $X$,  and top face $\mathcal{T}$  contained in $\{u=C_0\}$. The side $\mathcal{S}$ is obtained by flowing $\partial\mathcal{B}$ along $X$ until it reaches $\{u=C_0\}$. 
{Note that if $C_0$ is chosen sufficiently large, we can also ensure that the leaves of $\mathfrak F$ intersect the side $\mathcal S$ transversely.}

Since $M^n$ has only one end, $\mathcal{C}$ is compact. To prove \Cref{trivial top}, it suffices to show that $\mathcal{C}\cong \mathcal{B}\times[-C_0,C_0]$. 
Now we restrict our discussion to $\mathcal{C}$ for the remainder of the appendix and continue to  denote the restricted foliation $\mathfrak{F}|_\mathcal{C}$ as $\mathfrak{F}$. 
Let $\tilde{U}$ be the subset of $\mathcal{C}$ such that for any $p\in\tilde U$, the leaf $\mathfrak{L}_p$ through $p$
is compact and diffeomorphic to $\mathcal{B}$. 
\begin{lemma}\label{open}
    For any compact leaf $\mathfrak{L}$ in $\mathcal{C}$, there exists a tubular neighborhood $V$ such that $V\cong \mathcal{C}\times (-\delta,\delta)$ and each leaf in $V$ is compact.
\end{lemma}
\begin{proof} 
 We can construct such a neighborhood $V$ by flowing $X$ for a short time. A more general argument is stated in \cite[p.73, Lemma 6]{Camacho-Neto}.  
\end{proof}
Here is a useful concept in foliation theory which we will use in the proof.
\begin{definition}
    Let $S$ be a subset of $\mathcal{C}$, the saturation $\operatorname{sat}(S)$ is the union of all leaves $\mathfrak{L}$ that intersect with $S$. 
    We call the set $S$ saturated if $\operatorname{sat}(S)=S$. 
\end{definition}
The following lemma from \cite[p.47, Theorem 1 and 2]{Camacho-Neto} will be used later in the proof of Theorem \ref{trivial top}.
\begin{lemma} \label{sat op}
The following two statements hold:
    \begin{enumerate}
        \item If $S$ is an open subset of $\mathcal{C}$, then $\operatorname{sat}(S)$ is also open in $\mathcal{C}$. 
        \item Let $A$ be a saturated subset of $\mathcal{C}$. Then  $\partial^* A:=\partial A\setminus (\partial \mathcal{C}\cap A)$ is also saturated.
    \end{enumerate}
\end{lemma}
\begin{proof}
    (1) Let $w$ be a point in $\operatorname{sat} (S)$. Then there exists a leaf $\mathfrak{L}$ containing $w$ and another point $w_0\in S$. Let $\mathcal{L}$ be a connected open subset in $\mathfrak{L}$ containing $w$ and $w_0$, such that its closure $\overline{\mathcal{L}}$ is compact. Since $S$ is open,  flowing $\mathcal{L}$ by $X$ for a sufficiently short time produces an open neighborhood $V_0$ of $\mathfrak{L}$ such that each leaf in $V_0$ intersects $S$. Therefore, $V_0\subset \operatorname{sat}(S)$, which implies that $\operatorname{S}$ is open.  

   (2) Let $\mathring{A}$ be the largest open subset in $\mathcal{C}$ contained in $A$. Note that $\mathring{A}$ might not open in $M^n$, although it is open in $\mathcal{C}$. According to (1), we have $\operatorname{sat}(\mathring{A})$ is open in $\mathcal{C}$, and since $A$ is saturated, we have $\operatorname{sat}(\mathring{A})\subset A$. Thus, $\operatorname{sat}(\mathring{A})=\mathring{A}$, meaning that $\mathring{A}$ is saturated. 

Let $B=\mathcal{C}\setminus A$. By a similar argument, we can also show that $\mathring{B}$ is  saturated. Therefore, $\partial^*A=\mathcal{C}\setminus(\mathring{A}\cup \mathring{B})$ is also saturated.
\end{proof}

Let $U$ be the connected component of $\tilde{U}$ containing $\mathcal{B}$. Then $U$ is an open subset of $\mathcal{C}$. We will show that $\partial U=\partial\mathcal{C}$, which will imply $U=\mathcal{C}$.

\begin{lemma}\label{LA5}
  The leaves in $\mathcal{C}$ are compact and diffeomorphic to $\mathcal{B}$, i.e., $\partial^*U=\partial U\setminus (\partial \mathcal{C}\cap U)=\emptyset$.
\end{lemma}

\begin{figure}
    \begin{picture}(0,0)
     \put(50,20){$\mathcal{B}$} 
     \put(50,220){$\mathcal{T}$}
     \put(36,142){$\tilde{\mathfrak{L}}$}
     \put(78,137){\textcolor{blue}{$I$}}
     \put(95,117){\textcolor{red}{$p$}}
    \end{picture}
\includegraphics[width=0.35\textwidth]{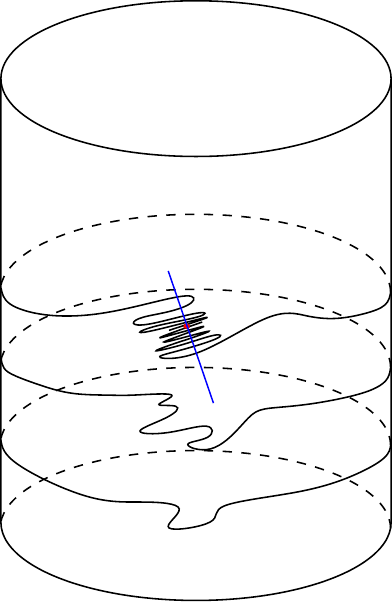}
    \caption{
       $\mathcal{T}$ and $\mathcal{B}$ are the top and base of the cylinder. 
    $I$ is the blue line segment passing through the red accumulation point $p$ of the noncompact leaf $\tilde{\mathfrak{L}}$.}
    \label{F:topology}
\end{figure}
\begin{proof}
Suppose  $\partial^*U\neq\emptyset$.
    According to \Cref{sat op}, $\partial^*U$ is saturated by $\mathfrak{F}$. We first show that all leaves contained in $\partial^* U$ are compact. 

    Suppose that $\tilde{\mathfrak{L}}\subset \partial^*U$ is a noncompact leaf. Since $\mathcal{C}$ is compact, there exists an accumulation point $p\in \mathcal{C}\setminus\tilde{\mathfrak{L}}$ of $\tilde{\mathfrak{L}}$.  
    %Note that since the foliation $\mathfrak{F}$ intersects $\partial \mathcal{C}$ transversely, we have $p\notin\partial \mathcal{C}$.
    Let $I$ be a compact line segment passing through $p$ which is an integral curve of $X$. 
    Notice that when restricting $\tilde{\mathfrak{L}}$ to a small neighborhood $V_p$ of $p$, we find that $\tilde{\mathfrak{L}}$ contains infinite many leaves in $V_p$ as $p$ is a limit point of $\tilde{\mathfrak{L}}$. Thus, $\tilde{\mathfrak{L}}\cap I$ is infinite. Considering the local neighborhood of each point in $\tilde{\mathfrak{L}}\cap I$, we have $U\cap I=\cup_{j\in \mathbb{N}}I_j$, where $I_j$ is a connected open segment in $I$. 

    Next, fixing $j$, we will show that for each $j$, $\operatorname{sat}(I_j)=U$. 
    \Cref{open} implies that $\operatorname{sat}(I_j)$ is an open subset in $U$. 
    Thus, it suffices to show that $\operatorname{sat}(I_j)$ is closed in $U$.  Let $\{p_i\}_{i=1}^\infty$ be a sequence of points in $ \operatorname{sat}(I_j)$ and $\lim_{i\to\infty}p_i=p_\infty\in U$. 
    Denote with $\mathfrak{L}_{p_\infty}\subset U$  as the leaf containing $p_\infty$. 
    Since  $\mathfrak{L}_{p_\infty}$ is contained in $U$, $\mathfrak{L}_{p_\infty}$ is compact by definition of $U$. Therefore, for sufficiently small $\delta>0$, there exists a saturated neighborhood $V_{p_\infty}^\delta$ of $\mathfrak{L}_{p_\infty}$ and we have $V_{p_\infty}^\delta\cong \mathfrak{L}_{p_\infty}\times (-\delta,\delta)$. 
    Denote with $\pi_\delta$ the retraction from $V_{p_\infty}^\delta$ to $\mathfrak{L}_{p_\infty}$.   
    Note that $p_i\in V_{p_\infty}^\delta\cap \operatorname{sat}(I_j)$ for large $i$, and there exist leaves contained in  $\operatorname{sat}(I_j)\cap V_{p_\infty}^\delta$.
    Hence, it follows that $I_j\cap V^\delta_{p_\infty}\neq \emptyset$. 
    Let $y\in I_j\cap V^\delta_{p_\infty}$. Then $\pi_\delta^{-1}(\pi_\delta(y))\subset I_j$ since $I_j$ is a connected component of $U\cap I$. 
    Therefore, $I_j$ intersects with all leaves in $V^\delta_{p_\infty}$ which is a contradiction. Thus, $\operatorname{sat}(I_j)=U$. Hence, every leaf in $U$ intersects $I$ infinitely many times. However, it is impossible, since one can choose $I$ sufficiently small so that $\mathcal{B}\cap I=\emptyset$. Thus, all leaves in $\partial^*U$ are compact. 

   Let $\mathfrak{L}'\subset \partial^*U$ be a compact leaf, and let $W$ be a tubular neighborhood of $\mathfrak{L}'$ such that $W\cong \mathfrak{L}'\times (-\delta',\delta')$ and $\partial^*W:=\partial W\setminus(W\cap \partial \mathcal{C})$ is compact.  We will show that there exists a leaf in $U\cap W$.
   If not, every leaf in $U$ which intersects with $W$ would also intersect with $\partial^* W$. Since $\mathfrak{L}'\subset \partial^*U$, there exists an open line  segment $\Gamma$ generated by $-X$ with starting point $q_\infty\in \mathfrak{L}'$ such that $\Gamma\subset U$. 
   Let $q_i\in \Gamma$ such that $q_i$ monotone converges to $q_\infty$ on $\Gamma$. 
   Denote $\mathfrak{L}_{q_i}$ as the leaf in $U$ containing $q_i$. Then $\mathfrak{L}_{q_i}\cap \partial^*W\neq \emptyset$. Thus, we can pick a point $\sigma_i\in\mathfrak{L}_{q_i}\cap \partial^*W$. Using the fact that $\partial^*W$ is compact, we obtain $\sigma_i\to \sigma_\infty$ after passing to a subsequence. 
Set $S_i=\overline{\operatorname{sat}(q_i,q_\infty)}$, where $(q_i,q_\infty)$ is the segment of $\Gamma$ between $q_i$ and $q_\infty$. Then $S_{i+1}\subset S_i$ and $S_i$ is compact and saturated. Note that $[q_i,q]\in S_i$, then $\sigma_j\in \mathfrak{L}_{q_j}\subset S_i$ for $j\ge i$. 
Thus, $\sigma_\infty\in S_i$, we have $\mathfrak{L}_{\sigma_\infty}\in S_i$, where $\mathfrak{L}_{\sigma_\infty}$ is the leaf containing $\sigma_\infty$. Since $S_i\in \overline{U}$,  we know that $\mathfrak{L}_{\sigma_\infty}$ is compact. Sine $\sigma_\infty\in \partial^* W$, we have $\sigma_\infty\notin \mathfrak{L}'$, and thus, $\mathfrak{L}_{\sigma_\infty}\neq \mathfrak{L}'$. Combining this with the fact $\mathfrak{L}_{\sigma_\infty}\subset S_i$, it follows that $\mathfrak{L}_{\sigma_\infty}$ intersects with $(q_i,q_\infty)$ for all $i$. Therefore, $q_\infty\in \mathfrak{L}_{\sigma_\infty}$ as $\mathfrak{L}_{\sigma_\infty}$ is compact which isa contradiction. Consequently, there exists a leaf $\hat{\mathfrak{L}}$ in $U\cap W$. 

Notice that there is a retraction $\pi: W\to \mathfrak{L}'$, and $\pi|_{\hat{\mathfrak{L}}}:\hat{\mathfrak{L}}\to \mathfrak{L}'$ is a covering map. Since $\hat{\mathfrak{L}}\subset U$, $\hat{\mathfrak{L}}$ is diffeomorphic to $\mathcal{B}$ and intersects $\partial \mathcal{C}$ transversely. Thus, $\mathfrak{L}'$ is diffeomorphic to $\mathcal{B}$. Hence, $\mathfrak{L}'$ has nonempty boundary, and $\mathfrak{L}'$ also intersects $\partial \mathcal{C}$ transversely, i.e., $\mathfrak{L}'\in U$. 
It follows that $U=\mathcal{C}$. 
\end{proof}
\begin{lemma}
    $\mathcal{C}$ is topologically trivial, i.e., $\mathcal{C}\cong \mathcal{B}\times [0,1]$.
\end{lemma}
\begin{proof}
    Let $\Upsilon$ be the integral curve of $X$ in $\partial\mathcal{C}$ connecting $\mathcal{B}$ and $\mathcal{T}$. We will show that $\operatorname{sat}(\Upsilon)=\mathcal{C}$.

    First, for any leaf $\mathfrak{L}\subset\operatorname{sat}(\Upsilon)$, $\mathfrak{L}$ is compact, and, by \Cref{open}, there exists a tubular open neighborhood $V_{\mathfrak{L}}$ of $\mathfrak{L}$ and $V_\mathfrak{L}\subset\operatorname{sat}(\Upsilon)$. Thus, $\operatorname{sat}(\Upsilon)$ is open in $\mathcal{C}$. 

    Next, we will show that $\operatorname{sat}(\Upsilon)$ is closed in $\mathcal{C}$. Observe that $\{V_\mathfrak{L}| \mathfrak{L}\in \operatorname{sat}(\Upsilon)\}$ forms an open cover of $\Upsilon$ and we can choose $\delta$ in \Cref{open} small so that $\overline{V_\mathfrak{L}}\subset \operatorname{sat}(\Upsilon)$. Since $\Upsilon$ is compact, there exists a finite cover $\{V_\mathfrak{L}^1,\dots,V_\mathfrak{L}^j\}$ of $\Upsilon$. Moreover, $\operatorname{sat}(\Upsilon)=\cup_{i=1}^j V_\mathfrak{L}^i$. On the other hand, $\overline{V_\mathfrak{L}^i}$ is compact and contained in $\operatorname{sat}(\Upsilon)$. Thus, $\operatorname{sat}(\Upsilon)=\cup_{i=1}^j \overline{V_\mathfrak{L}^i}$ is closed in $\mathcal{C}$.

Now we need to show that $\Upsilon$ intersects each leaf only once. 
Let $ \Upsilon_i $ be the subset of $ \Upsilon $ consisting of points $ x $ for which the leaf $ \mathfrak{L}_x $ containing $ x $ intersects $ \Upsilon $ exactly $ i $ times. 
By \Cref{open}, $ \Upsilon_i $ is open. Thus, $\Upsilon=\cup_i \Upsilon_i$ is a union of disjoint open sets. Since $\Upsilon$ is connected, we have $\Upsilon=\Upsilon_i$ for some $i\in \mathbb{N}$. Note that $\mathcal{B}\cap \Upsilon$ only contains one point, it follows that $\Upsilon=\Upsilon_1$. 
    
    Therefore, $\mathcal{C}\cong \mathcal{B}\times \Upsilon$ is topologically trivial.
\end{proof}
Hence, we conclude that $M^n$ is diffeomorphic to $\mathfrak{L}\times \R\cong\R^n$, thereby proving \Cref{trivial top}. 

As mentioned previously, the trivial topology of $\mathcal{C}$ is also a direct consequence of the Reeb's stability theorem with boundary \cite[Theorem 3.1, p.112]{Godbillon} which, for convenience, we state below:

\begin{theorem} \label{Reeb thm with boundary}
    Let $ \mathfrak{F} $ be a codimension-1 foliation of class $ C^r $ (with $ r \geq 1 $) on a compact connected manifold $ \Omega^n $, which is transverse or tangential to the (possibly empty) boundary of $ \Omega^n $. %if it is non-empty. 
    If $ \mathfrak{F} $ has a compact leaf $ \mathfrak{L} $ whose fundamental group is finite, then all its leaves are compact and have finite fundamental group as well. 
    Furthermore, if $ \mathfrak{F} $ is transversely orientable, this result still holds if we only assume that the leaf $ \mathfrak{L} $ has vanishing first Betti number. %a first Betti number equal to zero; 
    In this case, $ \mathfrak{F} $ is a fibration of $\Omega^n$ over the circle $ S^1 $ or the interval $ [0, 1] $.
\end{theorem}


\begin{thebibliography}{99}

\bibitem{AazamiCederbaum}
A.B. Aazami, C. Cederbaum, and C. Roche, \emph{Exact parallel waves in general relativity}, General Relativity and Gravitation, \textbf{55} (2023), no. 2, 40.

\bibitem{LIGO}
B. P. Abbott et al. (LIGO Scientific Collaboration and Virgo Collaboration), \emph{Observation of Gravitational Waves from a Binary Black Hole Merger}, Phys. Rev. Lett., \textbf{116} (2016), 061102.

\bibitem{Aichelburg}
P. C. Aichelburg and H. Balasin, \emph{Curvature without metric: the Penrose construction for half-flat pp-waves}, AVS Quantum Science 4, no. 2 (2022).

\bibitem{Araneda}
B. Araneda, \emph{Parallel spinors, pp-waves, and gravitational perturbations},  Classical and Quantum Gravity 40, no. 2 (2022): 025006.


\bibitem{harmonicspinor}
C. B\"{a}r, \textit{Metrics with harmonic spinors}, Geom. Funct. Anal., \textbf{6} (1996), 899--942.

\bibitem{BrendleChow}
C. Bär, S. Brendle, T.-K. A. Chow, and  B. Hanke, \textit{Rigidity of convex polytopes under the dominant energy condition}, arXiv preprint arXiv:2304.04145 (2023).

\bibitem{BGM}
C. Bär, P. Gauduchon, and A. Moroianu, \textit{Generalized cylinders in semi-Riemannian and spin geometry}, Mathematische Zeitschrift, \textbf{249} (2005), no. 3, 545--580.

\bibitem{BartnikChrusciel}
R. Bartnik, and P. Chruściel, \emph{Boundary value problems for Dirac-type equations}, Journal für die reine und angewandte Mathematik, \textbf{2005} (2005), no. 579, 13--73.

\bibitem{BeigChrusciel}
R. Beig, and P. Chru\'{s}ciel, \emph{Killing vectors in asymptotically flat space-times. I. Asymptotically translational Killing vectors and the rigid positive energy theorem}, J. Math. Phys., \textbf{37} (1996), no. 4, 1939--1961.

\bibitem{Blau} 
M. Blau, \emph{Lecture notes on general relativity}, Bern: Albert Einstein Center for Fundamental Physics, 2011.

\bibitem{Blau2}
M. Blau, D. Frank, and S. Weiss, \emph{Fermi coordinates and Penrose limits}, Classical and Quantum Gravity, \textbf{23} (2006), no. 11, 3993.

\bibitem{Blau3}
M. Blau, \emph{Plane Waves and Penrose Limits}, Lecture Notes, Bern: Albert Einstein Center for Fundamental Physics, 2011.

\bibitem{Bourguignon2015}
J.-P. Bourguignon, O. Hijazi, J.-L. Milhorat, A. Moroianu, and S. Moroianu, \emph{A Spinorial Approach to Riemannian and Conformal Geometry}, EMS Monogr. Math., Z\"{u}rich: European Mathematical Society (EMS), 2015.

\bibitem{BHKKZ} 
H. Bray, S. Hirsch, D. Kazaras, M. Khuri, and Y. Zhang, \textit{Spacetime harmonic functions and applications to mass}, Perspectives in scalar curvature. Vol. 2, 593–639, 2023.



\bibitem{Brinkmann1925}
H. Brinkmann, \emph{Einstein Spaces which are Mapped Conformally on Each Other}, Mathematische Annalen, \textbf{94} (1925), 119--145.

\bibitem{Bryant}
R. L. Bryant, \emph{Pseudo-Riemannian metrics with parallel spinor fields and vanishing Ricci tensor}, arXiv:math/0004073v1 [math.DG], 11 Apr 2000.


\bibitem{Camacho-Neto}
C. Camacho and A. Neto, \emph{Geometric theory of foliations}, Springer Science \& Business Media, 2013.


\bibitem{ChruscielMaerten}
P. Chru\'{s}ciel and D. Maerten, \textit{Killing vectors in asymptotically flat space-times. II. Asymptotically translational Killing vectors and the rigid positive energy theorem in higher dimensions}, J. Math. Phys., \textbf{47} (2006), no. 2, 022502, 10.

\bibitem{Eichmair}
M. Eichmair, \textit{The Jang equation reduction of the spacetime positive energy theorem in dimensions less than eight}, Comm. Math. Phys., \textbf{319} (2013), no. 3, 575-593.


\bibitem{EHLS}
M. Eichmair, L.-H. Huang, D. Lee, and R. Schoen, \textit{The spacetime positive mass theorem in dimensions less than eight}, J. Eur. Math. Soc. (JEMS), \textbf{18} (2016), no. 1, 83-121.

\bibitem{Einstein1}
A. Einstein, \emph{Näherungsweise Integration der Feldgleichungen der Gravitation}, 
Königlich Preußische Akademie der Wissenschaften (Berlin). Sitzungsberichte, (1916b), 688--696.

\bibitem{Einstein2}
A. Einstein, \emph{Über Gravitationswellen}, 
Königlich Preußische Akademie der Wissenschaften (Berlin). Sitzungsberichte, (1918), 154--167.


\bibitem{Hawking}
G. Gibbons, S. Hawking, G. Horowitz, and M. Perry, \emph{Positive Mass Theorems for Black Holes}, Comm. Math. Phys., \textbf{1983}.

\bibitem{Godbillon}
C. Godbillon, \textit{Feuilletages: études géométriques}, volume 98 of Prog. Math., Basel etc.: Birkhäuser Verlag, 1991.

\bibitem{HirschJangZhang}
S. Hirsch, H. C. Jang, and Y. Zhang, \emph{Rigidity of Asymptotically Hyperboloidal Initial Data Sets with Vanishing Mass}, preprint, arXiv:2411.07357, 2024.

\bibitem{HKK}
S. Hirsch, D. Kazaras, and M. Khuri, \textit{Spacetime Harmonic Functions and the Mass of 3-Dimensional Asymptotically Flat Initial Data for the Einstein Equations}, Journal of Differential Geometry 122, no. 2 (2022): 223-258.

\bibitem{HirschZhang}
S. Hirsch and Y. Zhang, \textit{The case of equality for the spacetime positive mass theorem}, J. Geom. Anal., \textbf{2023}.

\bibitem{HuangLee}
L.-H. Huang and D. Lee, \textit{Equality in the spacetime positive mass theorem}, Comm. Math. Phys., \textbf{2020}, 1-29.

\bibitem{HuangLee2}
L.-H. Huang and D. Lee, \textit{Equality in the spacetime positive mass theorem II}, Calculus of Variations and Partial Differential Equations 64, no. 3 (2025): 1-16.

\bibitem{HuangLee3}
L.-H. Huang and D. Lee, \textit{Bartnik mass minimizing initial data sets and improvability of the dominant energy scalar}, Journal of Differential Geometry 126, no. 2 (2024): 741-800.

\bibitem{LawsonMichelson}
H.B. Lawson, and M.-L. Michelsohn, \emph{Spin Geometry}, (PMS-38), Volume 38, Princeton University Press, 2016.

\bibitem{Lee}
D. Lee, \textit{Geometric Relativity}, Graduate Studies in Mathematics, Volume \textbf{201}, American Mathematical Society, 2019.

\bibitem{Morrey}
C. B. Morrey, \emph{Multiple Integrals in the Calculus of Variations}, Grundlehren Series, Vol. 130, Springer-Verlag, Berlin, 1966.

\bibitem{Penrose}
R. Penrose, \emph{Any space-time has a plane wave as a limit}, in: Differential Geometry and Relativity: A Volume in Honour of André Lichnerowicz on His 60th Birthday, Dordrecht: Springer Netherlands, 1976, pp. 271-275.

\bibitem{SY1}
R. Schoen, and S.-T. Yau, \textit{On the proof of the positive mass conjecture in general relativity}, Comm. Math. Phys., \textbf{65} (1979), no. 1, 45-76.

\bibitem{SY2}
R. Schoen, and S.-T. Yau, \textit{Proof of the positive mass theorem II}, Comm. Math. Phys., \textbf{79} (1981), 231-260.

\bibitem{SY3}
R. Schoen and S.-T. Yau, \textit{Positive scalar curvature and minimal hypersurface singularities}, Surveys in Differential Geometry, \textbf{24}(1), 2019, pp. 441-480.

\bibitem{Witten}
E. Witten, \textit{A simple proof of the positive energy theorem}, Comm. Math. Phys., \textbf{80} (1981), no. 3, 381-402.

\bibitem{Yip}
P.F. Yip, \emph{A strictly-positive mass theorem}, Comm. Math. Phys., \textbf{108} (1987), no. 4, 653-665.

\end{thebibliography}
\end{document}